\documentclass{lmcs}
\pdfoutput=1

\usepackage{lastpage}
\lmcsdoi{17}{3}{26}
\lmcsheading{}{\pageref{LastPage}}{}{}%
{Nov.~28,~2020}{Sep.~17,~2021}{}

\usepackage{amsmath,amssymb}
\usepackage[utf8]{inputenc}
\usepackage[T1]{fontenc}
\usepackage{comment}
\usepackage{booktabs}
\usepackage{mathpartir}

\usetikzlibrary{arrows,shapes,calc,matrix,positioning, backgrounds}

\theoremstyle{plain}

\DeclareMathOperator{\card}{card}
\newcommand{\execs}{\mathtt{execs}}
\newcommand{\DEL}{\ensuremath{\mathtt{DEL}}}
\newcommand{\Nexts}{\textit{Nexts}}
\newcommand{\HO}{\mathcal{HO}}
\newcommand{\HOProd}{\mathtt{HOProd}}
\newcommand{\obliv}{\textit{obliv}}
\newcommand{\cons}{\textit{cons}}
\newcommand{\rint}[1]{[#1[} 

\newcommand{\combi}{\bigotimes}


\begin{document}

\title[Characterization and Derivation of Heard-Of Predicates]{Characterization and Derivation of Heard-Of Predicates for Asynchronous Message-Passing Models}

\author{Adam Shimi}
\author{Aurélie Hurault}
\author{Philippe Queinnec}

\address{IRIT --- Université de Toulouse, 2 rue Camichel, F-31000 Toulouse, France}
\email{\{adam.shimi,aurelie.hurault,philippe.queinnec\}@irit.fr}

\begin{abstract}
In distributed computing, multiple processes interact to solve a problem
together. The main model of interaction is the message-passing model, where processes
communicate by exchanging messages. Nevertheless, there are several
models varying along important dimensions: degree of synchrony,
kinds of faults, number of faults\ldots{} This variety is compounded by the lack of
a general formalism in which to abstract these models, for translating results from
one to the other.
One way to bring order to these models is to constrain them further, to communicate in
rounds. This is the setting of the Heard-Of model, which captures many models
through predicates on the messages sent in a round and received on time (at
this round or before, where the round is the one of the receiver).
Yet, it is not easy to define the
predicate that best captures a given operational model. The question is even
harder for the asynchronous case, as unbounded message delay means the
implementation of rounds must depend on details of the model.

This paper shows that characterising asynchronous models by heard-of
predicates is indeed meaningful. This characterization relies on the
introduction of delivered predicates, an intermediate abstraction between
the informal operational model and the heard-of predicates. Our approach splits
the problem into two steps: first extract the delivered model capturing the
informal model, and then characterize the heard-of predicates that can be
generated by this delivered model.
For the first part, we provide examples of delivered predicates, and an
approach to derive more. It uses the intuition that complex models are
a composition of simpler models. We thus define operations like
union, succession or repetition that make it easier to derive complex
delivered predicates from simple ones while retaining expressivity. For the
second part, we formalize and study strategies for when to change rounds.
Intuitively, the characterizing predicate of a model is the one generated by
a strategy that waits for as much messages as possible, without blocking
forever.
\end{abstract}

\maketitle

\section{Introduction}%
\label{sec:chap2Intro}

\subsection{Motivation}

Distributed computing studies how multiple processes can accomplish computational tasks
by interacting with each other. Various means of communication are traditionally studied;
we focus here on message-passing, where processes exchange messages. Yet,
message-passing models still abound: they might have various degrees of synchrony (how
much processes can drift of from each other in term of processing speed or communication),
different kinds of faults (processes crashing, processes crashing and restarting,
message loss, message corruption\ldots{}), different network topologies\ldots{}
Although some of these differences are quantitative, such as the number of faults,
others are qualitative, like the kinds of faults. This in turn means that
abstracting most message-passing models into one formal framework proves difficult.
Indeed, most works in the literature limit themselves to either a narrow subset of such models
that can easily be abstracted together, or consider a handful of models one by one.
Hence it is difficult to compare, organize and extend results from
the literature to other message-passing models.
Formal verification is also made significantly
harder by the need to derive a formal version of each model.

One step towards unifying a broad range of message-passing models is to
abstract them through the messages which will be received. A crashed process would thus
be considered as a silent process, for example. Yet this information --- which messages
can be received --- is hard to capture in a simple mathematical object. This is one reason
why these proposals further constrain communication to explicitly use rounds: each
process repeatedly broadcasts a message with its current round number, waits
for some messages bearing this round number, and changes round by computing
its next state and its next message. The distributed algorithm tells each process
how to change its state at the end of each round, and which message to broadcast
in the next round.
Then the model can be abstracted through which message will be received by which
process at each round, and each such possibility is
represented either through a function or a sequence of graphs.
At first glance, forcing rounds onto these models might seem like a strong constraint,
severely limiting the generality of this unification. But rounds are actually ubiquitous
in distributed computing: they are present in the complexity analyses of synchronous
message-passing models, in the fault-tolerant algorithms for asynchronous models,
in the algorithms using failure-detectors\ldots{}
Since the use of rounds for this abstraction doesn't force them to be either
synchronous or asynchronous, it captures both.

Combining abstraction through received messages and rounds yields
the Heard-Of model of Charron-Bost and Schiper~\cite{CharronBostHO}. It abstracts
message-passing models through heard-of predicates, predicates over heard-of collections.
A heard-of collection is a function/sequence of graphs capturing
which messages are received across rounds in one execution.
The Heard-Of model thus boils down the problem of comparing and ordering message-passing
models to the one of comparing the corresponding heard-of predicate. But which heard-of
predicate captures a given message-passing model?
Heard-of predicates are properties of the rounds that can be implemented on top of
the original model. Hence this question requires a system-level approach: looking for rules
for when to change rounds that ensure that no process is blocked forever at a round.

In the synchronous case, where there is a bound on communication delay, every implementation
of rounds used in the literature implements rounds by waiting for the bound. This ensures
both the progress of rounds and the reception of every message from the round at that round.
Since this gives the maximal information to each process at each round, this is
the implementation of rounds on which the most problems can be solved.
It thus makes sense to define the heard-of predicate characterizing a synchronous model as
the one capturing this implementation: it simply specifies which messages
can be lost or never sent, and have all the others received on the round they were sent.

On the other hand, by definition, asynchronous models lack any upper bound on communication
delays. This results in a lot of variations on how rounds are implemented, depending on
the model and the problem one is trying to solve.
Waiting so little that the rule doesn't block processes on
any asynchronous model results in trivial heard-of predicates that capture nothing
specific about the model --- for example the rule that always allows to change rounds.
Conversely, rules that are tailored to a specific model, like waiting for $n$ messages, might
block processes in other models (where there might be less than $n$ messages received at
a round).
Finding a corresponding heard-of predicate for an asynchronous model thus requires
a model-specific analysis, as well as some way to choose among the possible rules.
This paper proposes a formalisation of this question, as well as an approach
for answering it for concrete message-passing models.

\subsection{Overview}%
\label{subsec:chap2Overview}

As mentioned above, abstracting message-passing models under one formalism would significantly
help with comparing results across models and formally verifying them. The Heard-Of model
provides such an abstraction, but only if we can define and compute a corresponding
heard-of predicate for asynchronous message-passing models, which is the hardest case due to
lack of an upper bound on communication delay.
We compute this heard-of predicate in two steps, as shown in
Figure~\ref{modelBase}.
We start with the operational model, derive a ``delivered predicate'', and
then find the heard-of predicates that can be implemented by some rule for changing rounds (called a strategy)
for this specific delivered predicate.
Among such predicates, we propose a criterion to choose the one characterizing
the asynchronous message-passing model.

\begin{figure}[t]
  \centering
  \begin{tikzpicture}
    \tikzstyle{link}=[->,>=stealth, line width=.45mm]
        \node[draw, circle, align=center, text width=2cm, inner sep=0pt, minimum size=1cm ] (start) at (-4,0) {\footnotesize Operational Model};
        \node[draw, circle, align=center, text width=2cm, inner sep=0pt, minimum size=1cm ] (mid) at (0,0) {\footnotesize Delivered Predicate};
        \node[draw, circle, align=center, text width=2cm, inner sep=0pt, minimum size=1cm ] (end) at (4,0) {\footnotesize Heard-Of Predicate};

        \draw[link] (start) to node[label={\footnotesize Rounds}] {} (mid);
        \draw[link] (mid) to node[label={\footnotesize Asynchrony}] {} (end);
  \end{tikzpicture}
  \caption{From Classical Model to Heard-Of Predicate}\label{modelBase}
\end{figure}
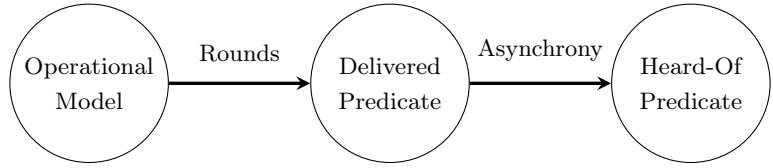

The first step goes from the original asynchronous message-passing model to
a delivered predicate, an abstraction that we introduce.
A delivered predicate captures the messages that are eventually delivered from each round $r$,
without considering the round of delivery. This is to contrast with heard-of
collections, that only capture messages tagged by $r$ if
delivered at the receiver when its local round counter is $\leq r$.
A delivered predicate  makes the original model formal
in the form of a delivered predicate, but it avoids dealing
with the main issue for getting a heard-of predicate: asynchrony.
Because the round of delivery is not considered in a delivered predicate,
computing the corresponding one for our original model does not require
a strategy for when to change rounds. Intuitively, the delivered predicate
of a model is the heard-of predicate of the same model if it was synchronous,
which are relatively straightforward to define and
compute.


The second step goes from the delivered predicate to the heard-of predicate
corresponding to the original model. To do so, we define strategies: rules
that tell when processes can change rounds. The main constraint on such strategies
is to never block a process at a round forever. As mentioned above, the difficulty here
comes from the fact that there are in general many different strategies for implementing
rounds that don't block forever. Which of these should be the corresponding
heard-of predicate?
Our answer relies on the following intuition: the predicate precisely capturing
the asynchronous model is the predicate satisfied by the fewer heard-of collections,
among the set of predicates that can be implemented on top of the original model.
Formally, this translates to being included in
every other heard-of predicate that can be implemented
(considering predicates as sets in the standard way).
Such a predicate, if it exists, intuitively constrains communication the most,
by allowing fewer possibilities for which messages are heard by who at which round.
This choice stems from the relevance of uncertainty for distributed computability.
A distributed algorithm has to give the correct answer independently of the specific
scheduling, the specific failures, and every other source of uncertainty and non-determinism
specified by the model. This means that if everything
that could happen according to model $M_1$ could also happen according to model $M_2$,
then every correct algorithm for $M_2$ is also correct for $M_1$. In the Heard-Of model,
this means that for two heard-of predicates $HO_1$ and $HO_2$ such that
$HO_1 \subseteq HO_2$, then every correct algorithm for $HO_2$ is also
correct on $HO_1$.


With this formalization in hand, what is left is a way to compute the
resulting predicate and prove that it is indeed the strongest.
This problem becomes tractable by introducing operations on predicates (delivered
and heard-of) and strategies.
Operations capture the intuition that it's often easier to build complex models
by composing simple ones together. Hence if the original model can be framed as such
a composition, its delivered predicate can similarly be constructed from the
delivered predicates of the building blocks, thanks to the operations we define.
For some families of strategies (strategies that only depend
on some limited part of the local state of a process, here the messages of the current round or
the messages of past and current rounds respectively), the strategy that
implements the heard-of predicate corresponding to the full model can be
built from the strategies of the building blocks using analogous operations on strategies
that we define. Then the heard-of predicate of this built strategy is linked
with the composition of the heard-of predicates implemented by the building block strategies.

\subsection{Contributions}

The contributions of this article are the following:

\begin{itemize}
  \item The definition of delivered predicates and strategies, in Section~\ref{sec:chap2defPdel}.
  \item Operations on delivered predicates and strategies, to build complex predicates, in Section~\ref{sec:chap2buildingDEL}.
  \item The formalization of the derivation of heard-of predicates
    from a delivered predicate and a strategy, in Section~\ref{sec:chap2fromDelToHO}.
    This comes with a complete example: the asynchronous message-passing model
    with reliable communication and at most $F$ permanent crashes.
  \item The study of oblivious strategies, the strategies only looking
    at messages for the current round, in Section~\ref{sec:chap2obliv}.
    We provide a technique to extract a strategy dominating the oblivious
    strategies of the complex predicate from the strategies of its building blocks;
    exact computations of the generated heard-of predicates;
    and a sufficient condition on the building blocks for the
    result of the operations to be dominated by an oblivious strategy.
  \item The study of conservative strategies, the strategies
    looking at all messages from previous and current round, as well
    as the round number,
    in Section~\ref{sec:chap2conserv}.
    We provide
    a technique to extract a strategy dominating the conservative
    strategies of the complex predicate
    from the strategies of its building blocks;
    upper bounds on the generated heard-of predicates;
    and a sufficient condition on the building blocks for the
    result of the operations to be dominated by a conservative strategy.
  \item A preliminary exploration of strategies using messages
    from future rounds, and an extended example where these strategies
    build stronger heard-of predicates than oblivious and conservative
    strategies, in Section~\ref{sec:chap2future}.
\end{itemize}

\subsection{Related Work}%
\label{subsec:related}

Rounds in message-passing algorithms date at least back
to their use by Arjomandi et al.~\cite{ArjomandiFirstRound} as
a synchronous abstraction of time complexity. Since then, they are
omnipresent in the literature. First, the number of rounds taken by a
distributed computation is a measure of its complexity.
Such round complexity was even developed into a full-fledged
analogous of classical complexity theory by Fraigniaud et
al.~\cite{FraigniaudComplexity}. Rounds also serve as stable
intervals in the dynamic network model championed by Kuhn and
Osham~\cite{KuhnDynamic}: each round corresponds
to a fixed communication graph, the dynamicity following
from possible changes in the graph from round to round. Finally,
many fault-tolerant algorithms are structured in rounds,
both synchronous~\cite{FisherAlgo} and asynchronous ones~\cite{ChandraCHT}.

Although we only consider message-passing models in this article, rounds are
also widely used in shared-memory models. A classic example is the structure
of executions underlying the algebraic topology approach pioneered by
Herlihy and Shavit~\cite{HerlihyTopology}, Saks and
Zaharoglou~\cite{SaksTopology}, and Borowsky and
Gafni~\cite{BorowskyTopology}.

Gafni~\cite{GafniRRFD} proposed a unification
of all versions of rounds with
the Round-by-Round Fault Detector abstraction,
a distributed module analogous to a failure detector which outputs
a set of suspected processes. In a system using RRFD, the end condition
of rounds is the reception of a message from every process not suspected
by the local RRFD module;
communication properties are then defined as predicates on the
output of RRFDs. Unfortunately, this approach fails to guarantee one property
that will prove necessary in the rest of this paper: the termination of
rounds.

Charron-Bost and Schiper~\cite{CharronBostHO} took a dual
approach to Gafni's work with the Heard-Of Model, by combining the concept of a fault
model where the only information is which
message arrives, from Santoro and Widmayer~\cite{SantoroLoss}. Instead of specifying
communication by predicates on a set of suspected processes, they
used heard-of predicates: predicates on a collection of
heard-of sets, one for each round $r$ and each process $j$, containing
every process from which $j$ received the message sent in
round $r$ before the end of this same round.
This conceptual shift brings two advantages: a purely
abstract characterization of message-passing models
and the assumption of infinitely many rounds, thus of round termination.

This model was put to use in many ways. Computability
and complexity results were proven: new algorithms for consensus
in the original paper by Charron-Bost and Schiper~\cite{CharronBostHO};
characterizations for consensus solvability by
Coulouma et al.~\cite{CouloumaConsensus} and Nowak et
al.~\cite{NowakTopoConsensus}; a characterization
for approximate consensus solvability by Charron-Bost et
al.~\cite{CharronBostApprox}; a study of $k$ set-agreement
by Biely et al.~\cite{BielyKSet}; and more.
The clean mathematical abstraction of the Heard-Of model also
works well with formal verification. The rounds provide
structure, and the reasoning can be less operational than
in many distributed computing abstractions. For instance, there exist
a verification with the proof assistant Isabelle/HOL of consensus algorithms
in Charron-Bost et al.~\cite{CharronBostHOL}, cutoff
bounds for the model checking of consensus algorithms
by Mari{\'{c}} et al.~\cite{MaricCutoff},
a DSL to write code following the structure of the Heard-Of model
and verify it with inductive invariants by Dr\u{a}goi et
al.~\cite{DragoiPsync}.

Yet, determining which model implements a given heard-of
predicate is an open question. As mentioned in
Mari{\'{c}}~\cite{MaricCutoff},
the only known works addressing it, one by Hutle and
Schiper~\cite{HutleComPredicates} and the other by Dr\u{a}goi et
al.~\cite{DragoiPsync}, both limit themselves
to specific predicates and partially synchronous system models.

\section{The Heard-Of Model}%
\label{sec:defho}

In the Heard-Of model of Charron-Bost and Schiper~\cite{CharronBostHO}, algorithms are defined by rounds, and an execution is an infinite sequence of rounds. At each round, a process broadcasts, receives, and does a local computation.
The Heard-Of model ensures communication-closedness~\cite{ElradDecomp}: in an algorithm,
processes at a given round only interact with processes at the same round ---
they only consider messages from this round. In the Heard-Of model, executions are necessarily infinite, and processes change rounds infinitely. Nevertheless, this does not prevent an algorithm from terminating, and, for instance to achieve consensus or election: the system reaches a configuration where the local state of each process does not change anymore.

The Heard-Of model constrains
communication through heard-of predicates, which are themselves
predicates on heard-of collections. These predicates play the role of
assumptions about synchrony, faults, network topology, and more.
Usually, heard-of collections are represented as functions
from a round $r$ and a process $p$
to a set of processes --- the processes
from which $p$ heard the message sent
at round $r$ before or during its own round $r$.

\begin{defi}[Heard-Of Collection and Predicate]%
\label{defHO}
  Let $\Pi$ a set of processes.
  A \textbf{heard-of collection} is an element $h$ of $(\mathbb{N}^* \times \Pi) \mapsto \mathcal{P}(\Pi)$.
  The \textbf{heard-of sets} of a heard-of collection are the outputs of this collection.
  A \textbf{heard-of predicate} $\HO$
  for $\Pi$ is a set of heard-of collections, that is an element of
  $\mathcal{P}((\mathbb{N}^* \times \Pi) \mapsto \mathcal{P}(\Pi))$.
\end{defi}

From another perspective, heard-of collections are infinite sequences of
communication graphs --- directed graphs which capture who hears from whom on time,
in that $q \in h(r,p) \iff (q,p)$ is an edge of the $r$-th communication graph.

\begin{defi}[Collection as a Sequence of Directed Graphs]
    Let $Graphs_{\Pi}$ be the set of directed graphs whose nodes are
    the elements of $\Pi$. Then $gr \in {(Graphs_{\Pi})}^{\omega}$
    is a \textbf{heard-of collection}.
    A function $h$ and a sequence $gr$ represent the same collection
    when $\forall r > 0, \forall p \in \Pi: h(r,p)
    = In_{gr[r]}(p)$, where $In(p)$ is the set of incoming neighbors of $p$.
\end{defi}

In general, which perspective to use in a theorem or a proof naturally follows
from the context. For example, $h[r]$ makes sense for a sequence
of directed graphs, while $h(r,p)$ makes sense for a function.

\section{Delivered Predicates: Rounds Without Timing}%
\label{sec:chap2defPdel}

Our concern is asynchronous message-passing models.
What makes the synchronous case easier than the asynchronous case boils down to the equivalence
between the messages that are received at all, and those that are received on time.
This cannot be replicated in the asynchronous case, as each asynchronous model
requires a different rule for which messages to wait for before changing round.
For example, in an asynchronous model with at most $F$ crashes, a process can wait for $n-F$ messages before changing round without risking waiting forever, as at least $n-F$ processes will never crash.
In the asynchronous model with at most $F+1$ crashes,
doing so will get processes blocked in some cases.
Nevertheless, it might be impossible to wait for all the messages that will be
delivered and not block forever. As explained before, in an
asynchronous model with at most $F$ crashes, process wait for $n-F$ messages
before changing round. If less than $F$ processes crash, not all the messages
will be waited for: a process may change round as soon as it has received
$n-F$ messages, and some messages will be received too late and therefore be ignored.

\subsection{Delivered Predicates}

Delivered sets are introduced to distinguish between the set of all
delivered messages (on time or late) and the messages delivered before
changing round (a heard-of set). A delivered collection is a sequence of
delivered sets, and a delivered predicate is a predicate on delivered
collections that defines which messages are received, ignoring changes of
rounds. Observe that a delivered predicate has the same formal definition as
Definition~\ref{defHO} of heard-of predicates --- only the interpretation
changes --, and the graph-based notation similarly applies.

\begin{defi}[Delivered Collection and Predicate]
  Let $\Pi$ a set of processes.
  A \textbf{delivered collection} is an element $c$ of $(\mathbb{N}^* \times \Pi) \mapsto \mathcal{P}(\Pi)$.
    The \textbf{delivered sets} of a delivered collection are the outputs of this collection.
    A \textbf{delivered predicate} $\DEL$
    for $\Pi$ is a set of delivered collections, that is an element of
    $\mathcal{P}((\mathbb{N}^* \times \Pi) \mapsto \mathcal{P}(\Pi))$.

\end{defi}

For examining the difference between heard-of and delivered collections, recall that we're considering the Heard-Of model at a
system level: we're implementing it. Let's take an execution of
some implementation (which needs to satisfy some constraints, defined later in Section~\ref{executions}):
a linear order of
emissions, receptions and changes of rounds (a step where the local round counter
is incremented) for each process. Then if each
process changes round infinitely often, there's a delivered collection $d$
and a heard-of collection $h$ corresponding to this execution --- just look
at which messages sent to $j$ tagged with round $r$ where received at all by $j$
(for $d$), and which were received when the round counter at $j$ was $\leq r$
(for $h$).
That is, for a round $r>0$ and processes $k,j \in \Pi$,
$k \in d(r,j)$ means that $j$ received
at some point the message of $k$ annotated by $r$.
On the other hand, $k \in h(r,j)$ means that
$j$ received  the message of $k$ annotated by $r$ while its round counter was $\leq r$.
Hence the heard-of collection extracted from this execution
captures which messages were waited for (and thus could be used at the algorithm
level --- that's not treated here), whereas the delivered collection extracted
from this execution captures which messages were received at all.
For example, the following scenario:
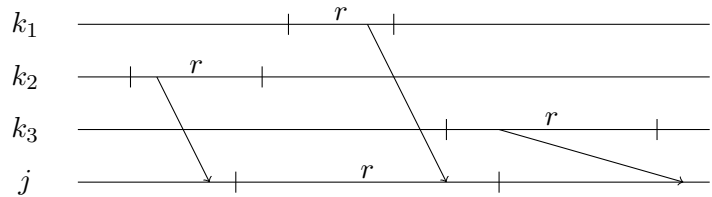
\begin{center}
  \begin{tikzpicture}[scale = 0.7]
    \tikzstyle{link}=[->,>=stealth, line width=.45mm]
        \node[align=center] at (0,0) {$k_1$};
        \node[align=center] at (0,-1) {$k_2$};
        \node[align=center] at (0,-2) {$k_3$};
        \node[align=center] at (0,-3) {$j$};

     	\draw[] (1,0) -- (13,0);
	\draw[] (1,-1) -- (13,-1);
	\draw[] (1,-2) -- (13,-2);
	\draw[] (1,-3) -- (13,-3);

	\draw[] (5,0.2) -- (5,-0.2);
	\draw[] (7,0.2) -- (7,-0.2);
	\node[align=center] at (6,0.2) { $r$};

	\draw[] (2,-0.8) -- (2,-1.2);
	\draw[] (4.5,-0.8) -- (4.5,-1.2);
	\node[align=center] at (3.25,-0.8) { $r$};

	\draw[] (8,-1.8) -- (8,-2.2);
	\draw[] (12,-1.8) -- (12,-2.2);
	\node[align=center] at (10,-1.8) { $r$};

	\draw[] (4,-2.8) -- (4,-3.2);
	\draw[] (9,-2.8) -- (9,-3.2);
	\node[align=center] at (6.5,-2.8) { $r$};

	\draw[->] (6.5,0) -- (8,-3);
	\draw[->] (2.5,-1) -- (3.5,-3);
	\draw[->] (9,-2) -- (12.5,-3);

  \end{tikzpicture}
\end{center}
leads to: $\{k_1,k_2,k_3\} \subseteq d(r,j)$ and $\{k_1,k_2\} \subseteq h(r,j)$.

To find the delivered predicate corresponding to an asynchronous model, the intuition
is to take the synchronous version of the model, and then take the heard-of predicate that would be
implemented by the rule for changing rounds in synchronous models.
This is the delivered predicate for the model. This
captures the strongest heard-of predicate that could be
implemented on top of this asynchronous model, if processes could
wait for all messages that will be delivered.
In general, they can't, since it requires knowing exactly what's happening
over the whole distributed system.
Nonetheless, the delivered predicate exists,
and it plays the role of an ideal to strive for.
The characterizing heard-of predicate of a model will be the closest
overapproximation of the delivered predicate that can actually be implemented.

\subsection{A Delivered Predicate for at Most \texorpdfstring{$F$}{F} Crashes}

As a first example, we consider the asynchronous model with reliable communication,
and at most $F$ crash failures (where crashes can happen at any point).

\begin{defi}[$\DEL^{crash}_{F}$]%
    \label{delCrash}
    The delivered predicate $\DEL^{crash}_{F}$ for the asynchronous model
    with reliable communication and at most $F$ permanent crashes $\triangleq$
    \[
    \left\{ c, \textit{a delivered collection}
    ~\middle|~
        \forall r > 0, \forall p \in \Pi:
        \begin{array}{ll}
            & \card(c(r,p)) \geq n-F\\
            \land & c(r+1,p) \subseteq K_c(r)\\
        \end{array}
        \right\},\]
        where $\ensuremath{K_c(r)}$ is the kernel of $c$ at $r$: $K_c(r) \triangleq \bigcap\limits_{p \in \Pi} c(r,p)$, and $\card$ is the cardinality function.

\end{defi}

Charron-Bost and
Schiper~\cite[Table 1]{CharronBostHO} define it as the heard-of predicate of
the synchronous version of this model. We now give an argument for why,
if you take the asynchronous model with reliable communication and at most $F$
permanent crashes, and implement communication-closed rounds in any way that
ensures an infinite number of rounds for every process, the messages received
will form a delivered collection of $\DEL^{crash}_F$.
In the other direction, every collection of $\DEL^{crash}_F$ captures
the messages received in an execution of the implementation of rounds on top
of the aforementioned asynchronous model.

\begin{itemize}
  \item Let $t$ be an execution of an implementation of communication-closed
    rounds on top of the asynchronous model above, with the condition of
    ensuring an infinite number of rounds.
    We consider that a crashed process is modelled as a silent process: a
    crashed process will still receive all messages after it crashes, but
    will never do anything else after it crashed. This is observationally
    indistinguishable for the other processes.
    Since every process that has not crashed
    broadcasts, this entails that every process will eventually hear
    the message from every non-crashed process at this round.
    Since there's at most $F$ crashes, that's at least $n-F$ messages
    per round. Hence $\card(c(r,j)) \geq n-F$ for every process $j$ and round~$r$.

    Also, if $p$ hears from $k$ at round $r+1$, then $k$ sent the message
    before crashing. This means $k$ did not crash at its own round $r$,
    and thus that the message it broadcast at that round $r$ was sent,
    and eventually received by all processes.
    Hence $c(r+1,j) \subseteq K_{\textit{c}}(r)$.
  \item Let $c$ be a collection such that
    $\forall r > 0, \forall p \in \Pi: \card(c(r,p)) \geq n-F
    \land c(r+1,p) \subseteq K_{\textit{c}}(r)$. This collection
    corresponds to the execution where the crashed processes are
    the ones that stop broadcasting. Because communication is
    reliable, $k \notin c(r,p)$ means that $k$
    never sent its message to $p$ tagged with $r$. From the model
    this means that it crashed during its broadcast at round $r$ or earlier.
    Each crash thus happens at the first round where the crashed process
    is not heard by everyone, after sending
    the messages that are actually received at this round.
\end{itemize}

\noindent
Later in this article, the heard-of predicate characterizing this delivered predicate
(the most constrained one) is derived.
We mention it here, as a comparison~\cite{CharronBostHO}:
\[ \HO_F  \triangleq \{ h, \textit{a heard-of collection} \mid
\forall r > 0, \forall p \in \Pi: \card(ho(r,p)) \geq n-F\}\]
The difference lies in the kernel condition:
$\DEL^{crash}_{F}$ ensures that if any message sent by $p$
at round $r$ is not eventually delivered, then no message will be delivered
from $p$ at rounds $>r$. Intuitively, $p$ not broadcasting means that it
crashed during round $r$ or earlier,
and that it will never send messages for the next
rounds. However, this is not maintained by $\HO_F$,
as the $n-F$ messages that are waited for are not necessarily
the same for each process. So $k$ might wait for a message from $p$ at round $r$,
but $j$ might receive at least $n-F$ messages at round $r$ without the one from $p$.
$j$ cannot conclude that $p$ has crashed,
because the message from $p$ might just be late.

\subsection{A Delivered Predicate for at Most \texorpdfstring{$L$}{L} Losses}

As a second example, we consider the asynchronous
model without crashes, but with at most $L$ messages lost in the whole execution
(of the system level implementation of rounds).

\begin{defi}[$\DEL^{loss}_{L}$]
    The delivered predicate $\DEL^{loss}_{L}$ for the asynchronous model
    without crashes, and with at most $L$ message losses
    $\triangleq\\
    \left\{ c, \textit{a delivered collection}
    ~\middle|~
        \sum\limits_{r > 0, p \in \Pi}(n - \card(c(r,p))) \leq L
    \right\}$.
\end{defi}

This one is not from Charron-Bost and Schiper~\cite{CharronBostHO}, but
we can apply the same reasoning as for the previous delivered predicate. Here
the sum counts the number of messages that are never delivered.
Since all the processes are correct, this corresponds to the number of lost messages.
For $L = 1$, the best known strategy (to our knowledge) implements the heard-of predicate
\[\{ h, \textit{a heard-of collection} \mid \forall r > 0, \sum\limits_{p \in \Pi} \card(\Pi \setminus ho(r,p)) \leq 1\}.\]
What is lost in implementing a heard-of predicate on top of the delivered predicate $\DEL^{loss}_1$
is that instead of losing only one message over the whole execution,
there might be one loss per round.
This is explained in Section~\ref{sec:chap2future}.

\section{Composing Delivered Predicates}%
\label{sec:chap2buildingDEL}

Finding the delivered predicate for a complex model is difficult. However,
simple models are relatively easy to characterize by a delivered predicate.
This motivates the following proposal to solve the plain part of Figure~\ref{modelDEL}:
composing simple delivered predicates to derive complex delivered predicates.
That way, there will be no need to define by hand the
delivered predicates of complex models.

\begin{figure}[t]
  \centering
  \begin{tikzpicture}
    \tikzstyle{link}=[->,>=stealth, line width=.45mm]
        \node[draw, circle, align=center, text width=2cm, inner sep=0pt, minimum size=1cm ] (start) at (-4,0) {\footnotesize Operational Model};
        \node[draw, circle, align=center, text width=2cm, inner sep=0pt, minimum size=1cm ] (mid) at (0,0) {\footnotesize Delivered Predicate};
        \node[dashed, draw, circle, align=center, text width=2cm, inner sep=0pt, minimum size=1cm ] (end) at (4,0) {\footnotesize Heard-Of Predicate};

        \draw[link,align=center] (start) to node[label={\footnotesize Rounds \\ \footnotesize (Operations)}] {} (mid);
        \draw[link,dashed] (mid) to node[label={\footnotesize Asynchrony}] {} (end);
  \end{tikzpicture}
  \caption{From Classical Model to Delivered Predicate}\label{modelDEL}
\end{figure}
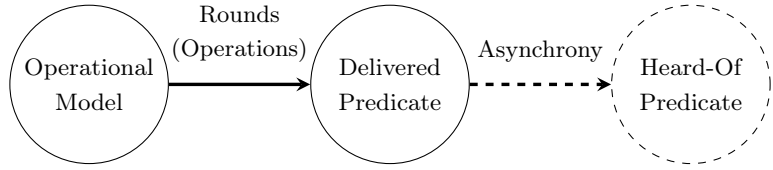

\subsection{Introductory Example}

Consider a system where one process might crash and may
or may not recover later on\footnote{If it does, we can assume that
its memory is intact and no messages received in the meantime are lost,
but that's not important for the system level implementation.}.
In some sense, this behavior is defined by having the delivered collections
for one possible crash that never recovers, and the delivered
collections for one possible crash that must recover.
This amounts to a union (or a disjunction);
we write it $\DEL^{canrecover}_1 \triangleq
\DEL^{crash}_1 \cup \DEL^{recover}_1$.
The first predicate of this union is $\DEL^{crash}_1$, the delivered
predicate for at most one crash that never recovers. Let's consider the second
one, $\DEL^{recover}_1$.
Intuitively, a process that can crash but must recover afterward is
described by the behavior of $\DEL^{crash}_1$ which is shifted to
the behavior of $\DEL^{total}$ (the predicate where all the messages
are delivered) after some time.
We call this the succession of these predicates,
and write it $\DEL^{recover}_1 \triangleq \DEL^{crash}_1 \leadsto \DEL^{total}$.
Finally, imagine adding to this system another permanent crash. The full
behavior is such that there might be one crashed process as constrained by
$\DEL^{crash}_1$, and another crashed process as constrained by
$\DEL^{canrecover}_1$. We call it the combination (or conjunction)
of these predicates, and write it $\DEL^{crash}_1 \combi \DEL^{canrecover}_1$.
The complete system is thus described by $\DEL^{crash}_1 \combi ((\DEL^{crash}_1
\leadsto \DEL^{total}) \cup \DEL^{crash}_1)$.

In the following, we also introduce an operator $\omega$ to express
repetition. For example, a system where, repeatedly, a process can crash and recover is
${(\DEL^{crash}_1 \leadsto \DEL^{total})}^\omega$.

\subsection{Operations on predicates}

Let's now formally define these operations.

\begin{defi}[Operations on predicates]%
    \label{defOpsPred}
    Let $P_1, P_2$ be two delivered or heard-of predicates.
    \begin{itemize}
      \item The \textbf{union} of $P_1$ and $P_2$ is $P_1 \cup P_2$.
      \item The \textbf{combination} $P_1 \combi P_2 \triangleq
        \{c_1 \combi c_2 \mid c_1 \in P_1, c_2 \in P_2 \}$,
        where if $c_1$ and $c_2$ are two collections,
        $\forall r > 0, \forall p \in \Pi: (c_1 \combi c_2)(r,p) =
        c_1(r,p) \cap c_2(r,p)$.
      \item The \textbf{succession} $P_1 \leadsto P_2 \triangleq
        \bigcup\limits_{c_1 \in P_1, c_2 \in P_2} c_1 \leadsto c_2$,\\
        with $c_1 \leadsto c_2 \triangleq \{ c \mid
        \exists r \geq 0 : c = c_1[1,r].c_2\}$ ($c_1[1,0]$ is the empty
        sequence).
      \item The \textbf{repetition} of $P_1$, ${(P_1)}^{\omega} \triangleq
        \{c \mid \exists {(c_i)}_{i \in \mathbb{N}^*},
        \exists {(r_i)}_{i \in \mathbb{N}^*}:
        r_1 = 0 \land
        \forall i \in \mathbb{N}^*:
        (c_i \in P_1 \land r_{i} < r_{i+1} \land
        c[r_i+1,r_{i+1}]=c_i[1,r_{i+1} - r_i]) \}$.
    \end{itemize}
\end{defi}

\noindent
The intuition behind these operations is the following:
\begin{itemize}
  \item The union of two delivered predicates is equivalent
    to an OR on the two communication behaviors. For example,
    the union of the delivered predicate for one crash at round $r$
    and of the one for one crash at round $r+1$ gives
    a predicate where there is either a crash at round $r$
    or a crash at round $r+1$.
  \item The combination of two behaviors takes
    every pair of collections, one from each predicate, and
    computes the intersection of the graphs at each round. Meaning, it
    adds the loss of messages from both, to get both behaviors
    at once.
    For example, the combination of the delivered predicate for one crash at round $r$
    and of the one for one crash at round $r+1$ gives
    a predicate where there is a crash at round $r$
    and a crash at round $r+1$.
    Observe that combining $\DEL^{crash}_1$ with itself gives
    $\DEL^{crash}_2$, the predicate with at most two crashes.
  \item For succession, the system
    starts with one behavior, then switches to another. The definition is
    such that if $r=0$, then no prefix of $c_1$ is used (the first
    behavior never happens), but the second one must always happen.
    For example, the succession of $\DEL^{crash}_1$ (one possible crash) with $\DEL^{total}$ (no crash) is a possible crash that recovers.
  \item Repetition is the next logical step after succession: instead of following
    one behavior with another, the same behavior is repeated again and again.
    For example, taking the repetition of at most one crash results
    in a potential infinite number of crash-and-restart,
    with the constraint of having at most one crashed process at any time.
\end{itemize}

\subsection{Basic blocks}

The usefulness of these operations comes from allowing the construction of interesting
predicates from few basic ones. Let's take a simple family
of basic blocks: $\DEL^{crash}_{1,r}$, the delivered predicate of the
model with at most one crash, at round $r$.

\begin{defi}[At most $1$ crash at round $r$]%
  \label{crash-round}
  $\DEL^{crash}_{1,r} \triangleq$
  \[\left\{
  c, \textit{a delivered collection}
  \middle| \exists~ \Sigma \subseteq \Pi:
  \begin{array}{ll}
    & \card(\Sigma) \geq n-1\\
    \land & \forall p \in \Pi
      \left(
      \begin{array}{lll}
        & \forall r' \in \rint{1,r}:
          & c(r',p) = \Pi\\
        \land & & c(r,p) \supseteq \Sigma\\
        \land & \forall r' > r:
          & c(r',p) = \Sigma\\
      \end{array}
      \right)
  \end{array}
  \right\}.\]
\end{defi}

In these predicates, before round $r$, every process receives every message.
At round $r$ a crash might happen, which means that processes only
receive messages from a subset $\Sigma$ of $\Pi$ of size $\card(\Pi)-1$ from round $r+1$
onwards. The subtlety at round $r$ is that the crashed process
(the only one in $\Pi \setminus \Sigma$) might crash while sending messages,
and thus might send messages to some processes and not others.

Another fundamental predicate is the total one: the predicate containing
a single collection $c_{total}$, the one where every process receives
every message at every round.

\begin{defi}[Total delivered predicate]%
  \label{total}
  $\DEL^{total} \triangleq \{c_{total}\}$, where
  $c_{total}$ is the collection defined by $\forall r > 0, \forall p \in \Pi:
  c(r,p) = \Pi$.
\end{defi}

Using these building blocks, many interesting and important delivered predicates
can be built, as shown in Table~\ref{tab:examples}.
For example, let's take $\DEL^{crash}_{1}$, the predicate with at most one crash.
If a crash happens, it happens at one specific round $r$.
$\DEL^{crash}_{1}$ is a disjunction
for all values of $r$ of the predicate with at most one crash at round $r$;
that is, the union of $\DEL^{crash}_{1,r}$ for all $r$.

\begin{table}[t]
\centering
\begin{tabular}{ll}
    Description & Expression\\\toprule
    At most 1 crash &
        $\DEL^{crash}_{1} =
        \bigcup\limits_{i=1}^{\infty} \DEL^{crash}_{1,i}$\\\midrule
    At most $F$ crashes &
        $\DEL^{crash}_{F} =
        \combi\limits_{j=1}^F \DEL^{crash}_1$\\\midrule
        At most 1 crash, which will restart &
        $\DEL^{recover}_{1} =
        \DEL^{crash}_1 \leadsto \DEL^{total}$\\\midrule
        At most $F$ crashes, which will restart &
        $\DEL^{recover}_{F} =
        \combi\limits_{j=1}^F \DEL^{recover}_{1}$\\\midrule
        At most $1$ crash, which can restart &
        $\DEL^{canrecover}_{1} =
        \DEL^{recover}_{1} \cup \DEL^{crash}_1$\\\midrule
        At most $F$ crashes, which can restart &
        $\DEL^{canrecover}_{F} =
        \combi\limits_{j=1}^F \DEL^{canrecover}_{1}$\\\midrule
    \begin{tabular}{l}
      No bound on crashes and restart,\\
      with only $1$ crash at a time
    \end{tabular}&
        $\DEL^{recovery}_1 =
        {(\DEL^{crash}_1)}^{\omega}$\\\midrule
    \begin{tabular}{l}
      No bound on crashes and restart,\\
      with max $F$ crashes at a time
    \end{tabular}&
        $\DEL^{recovery}_F =
        \combi\limits_{j=1}^F \DEL^{recovery}_1$\\\midrule
        At most $1$ crash, after round $r$ &
        $\DEL^{crash}_{1,\geq r} =
            \bigcup\limits_{i=r}^{\infty} \DEL^{crash}_{1,i}$\\\midrule
        At most $F$ crashes, after round $r$ &
        $\DEL^{crash}_{F,\geq r} =
            \bigcup\limits_{i=r}^{\infty} \DEL^{crash}_{F,i}$\\\midrule
    \begin{tabular}{l}
        At most $F$ crashes with no more than\\
        one per round\\
    \end{tabular}&
        $\DEL^{crash\neq}_{F} =
            \bigcup\limits_{i_1 \neq i_2 \neq \dots \neq i_F}
            \combi\limits_{j=1}^{F} \DEL^{crash}_{1,i_j}$\\\bottomrule
\end{tabular}
\caption{Delivered Predicates Built Using our Operations}%
\label{tab:examples}
\end{table}

\section{From Delivered Predicates to Heard-Of Predicates}%
\label{sec:chap2fromDelToHO}

After defining delivered predicates and discussing how to find and/or build them,
the next step is to study the heard-of predicates that can be implemented
over a given delivered predicate. This is the plain part of Figure~\ref{model}, which
works between two mathematical abstractions, and is formal.
To do so, we start by defining executions on the delivered predicate: traces of
the system's behavior that can be analysed. Next, we define strategies,
which capture the rules for changing rounds, and thus constrain possible executions
to those where changes of rounds happen only when allowed by the strategy.
If such a strategy never generates an execution
where some process is blocked forever at a round, it is called valid, and implements
a heard-of predicate (one collection per execution). Finding the heard-of predicate
corresponding to a given delivered predicate then boils down to defining a partial order
of valid strategies by way of their heard-of predicate, and taking the greatest element.

\begin{figure}[t]
\centering
\begin{tikzpicture}
  \tikzstyle{link}=[->,>=stealth, line width=.45mm]
      \node[dashed, draw, circle, align=center, text width=2cm, inner sep=0pt, minimum size=1cm ] (start) at (-4,0) {\footnotesize Operational Model};
      \node[draw, circle, align=center, text width=2cm, inner sep=0pt, minimum size=1cm ] (mid) at (0,0) {\footnotesize Delivered Predicate};
      \node[draw, circle, align=center, text width=2cm, inner sep=0pt, minimum size=1cm ] (end) at (4,0) {\footnotesize Heard-Of Predicate};

      \draw[link,dashed,align=center] (start) to node[label={\footnotesize Rounds \\ \footnotesize (Operations)}] {} (mid);
      \draw[link,align=center] (mid) to node[label={\footnotesize Asynchrony \\ \footnotesize (Strategies)}] {} (end);
\end{tikzpicture}
\caption{From Delivered Predicate to Heard-Of Predicate}\label{model}
\end{figure}
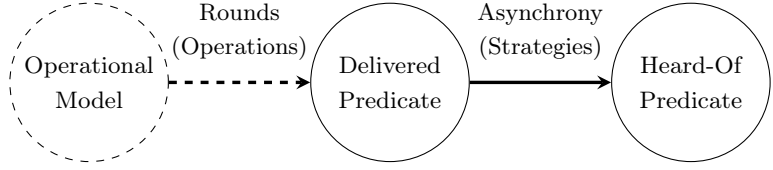

\subsection{Executions}%
\label{executions}

The theory of distributed computing relies on the concept of executions:
traces of a system's behavior, with enough detail to be formally analysed.
Here,
as we study the system-level implementation of the Heard-Of model, the executions
we consider are not executions of an algorithm solving a distributed
computing problem, but the executions of the implementation of a specific
heard-of predicate. Hence, these executions only track emissions, receptions
and changes of rounds. Because the content of each message is not important
for the implementation of rounds,
and we care about which messages will be received on time, the emissions
are implicit: as long as a process changed round $r-1$ times, it sent
its messages for round $r$ (which messages will depend on the delivered collection
used, as explained in a few paragraphs).
As for the local state of each process during this implementation,
it contains a local round counter and the set of received messages.

The last thing that is missing here is the implementation algorithm:
the rule that specifies when to change rounds. This is called
a strategy, and is defined below. First, we define executions
as sequences of events
that satisfy some basic constraints on the ordering of events. We
then constrain them by requiring the delivery of exactly the messages
from some delivered collection. The
introduction of strategies constrains them some more, so that the executions
allowed are the executions of an implementation of
rounds using this strategy.

Executions are infinite sequences of events,
either delivery
of messages ($deliver(r,k,j)$, which represents the delivery
at $j$ of the message from $k$ tagged with $r$), change to the next
round for some process $j$ ($next_j$), which also includes the broadcast for the
next round, or a deadlock ($stop$).
An execution must satisfy three rules:
no message is delivered before it is sent, no message
is delivered twice, and once there is a $stop$, the rest of
the sequence can only be $stop$.

\begin{defi}[Execution]%
  \label{defExec}
  Let $\Pi$ be a set of $n$ processes.
  Let the set of transitions $T = \{ \textit{next}_j \mid j \in \Pi \} \cup
  \{ \textit{deliver}(r,k,j) \mid r \in \mathbb{N}^* \land k,j \in \Pi\}
  \cup \{ stop \}$.
  Then, $t \in T^{\omega}$ is an \textbf{execution}
  $\triangleq$
  \begin{itemize}
      \item (Delivery after sending)\\
        $\forall i \in \mathbb{N}:
        t[i] = deliver(r,k,j) \implies
        \card(\{l \in \rint{0,i} \mid t[l] = next_k\}) \geq r-1$
      \item (Unique delivery)\\
        $\forall \langle r, k, j \rangle \in
        (\mathbb{N}^* \times \Pi \times \Pi):
        \card(\{i \in \mathbb{N} \mid t[i] = deliver(r,k,j)\}) \leq 1$
      \item (Once stopped, forever stopped)\\
        $\forall i \in \mathbb{N}:
        t[i] = stop \implies \forall l \geq i : t[l] = stop$
  \end{itemize}
\end{defi}

\noindent
Executions can be constrained by a delivered collection $c$: if $k$ changes round at least $r-1$ times in the execution,
then it sends all the messages tagged with $r$ to processes $j$ satisfying
$k \in c(r,j)$, and these messages are delivered in the execution.
Moreover, every delivery must be of such a message.
The executions of a delivered predicate are the executions of the collections of the predicate.

\begin{defi}[Execution of a delivered collection (and of a predicate)]%
  \label{defExecColl}
  Let \textit{c} be a delivered collection.
  Then, $\execs(c)$, the \textbf{executions of}
  \textit{c} $\triangleq$
  \[\left\{
      t, \textit{an execution} ~\middle|~
      \begin{array}{l}
          \forall \langle r, k, j \rangle
          \in \mathbb{N}^* \times \Pi \times \Pi:\\
          \quad(
              k \in c(r,j)
              \land \card(\{ i \in \mathbb{N} \mid t[i] = next_k\}) \geq r-1
          )\\
          \quad\iff\\
          \quad(
              \exists i \in \mathbb{N}: t[i] = deliver(r,k,j)
          )
      \end{array}
      \right\}.\]
  For a delivered predicate $\DEL$, we write
  $\execs(\DEL) = \bigcup\limits_{c \in \DEL} \execs(c)$.
\end{defi}

Definition~\ref{defExec} above casts behavior in term of changes to the system ---
the deliveries and changes of rounds. A dual perspective interprets behavior
as the sequence of successive states of the system. In a distributed system,
such states are the product of local states.
The local state of a process is the pair of its current
round and all the received messages up to this point\footnote{Recall that this is the
local state of the system-level implementation of rounds, not of the algorithm
running on top of the Heard-Of model. Hence, this doesn't constrain the internal
structure of the algorithms.}.
Notably, such a local
state doesn't contain the identity of the process. This is both because we
never need this identity, and because not dealing with it allows an easier
comparison of local states, since two distinct processes can have the same
local state. A message is represented
by a pair $\langle round, sender \rangle$ (instead of triplet like in deliver events,
because the receiver is implicit --- it's the process whose local state we're looking
at). For a state $q$, $q(r)$ is the set of peers from which the process (with state $q$)
has received a message tagged with round $r$.

\begin{defi}[Local State]%
  \label{defLocalState}
  Let $Q = \mathbb{N}^* \times \mathcal{P}(\mathbb{N}^* \times \Pi)$.
  Then $q \in Q$ is a \textbf{local state}.

  For $q = \langle r, mes \rangle$, we write $q.round$ for $r$,
  $q.mes$ for $mes$ and $\forall r' > 0: q(r') \triangleq
  \{k \in \Pi \mid \langle r', k \rangle \in q.mes\}$.

  Let $t$ be an execution, $p \in \Pi$ and $i \in \mathbb{N}$.
  Then the local state of $p$ in $t$ after the prefix of $t$ of size $i$ is
  $q_p^t[i] \triangleq
  \langle \card(\{ l < i \mid t[l] = next_p\})+1,
  \{\langle r, k \rangle \mid \exists l < i:
  t[l] = deliver(r,k,p)\} \rangle$.
\end{defi}

Notice that such executions do not allow a process to ``jump'' from say
round $5$ to round $9$ without passing by the rounds in-between. Indeed,
the Heard-Of model doesn't let processes decide when to change rounds:
processes specify only which messages to send depending on
the state, and what is the next state depending on the current state and
the received messages. So it makes sense for a system-level implementation
of heard-of predicates to do the same.
Nevertheless, the algorithm running on top of the Heard-Of model can ``jump'' rounds,
by not doing anything for a certain number of rounds.

\subsection{Strategies and Composition of Strategies}

An execution of a delivered collection where all processes change round
infinitely often defines a heard-of collection. This is done by looking, for
each round $r$ and process $p$, at the set of processes such that $p$
received their message tagged by $r$ when the round counter at $p$ was $\leq
r$. However, not all of the executions, as defined in~\ref{defExecColl},
ensure an infinite number of rounds for each process. For example, for a
delivered collection $c$, the execution where all messages from round $1$
are delivered according to $c$ (whatever the order) and then $stop$
transitions happen forever is an execution of $c$. Yet it blocks all
processes at round $1$ forever.
Strategies are introduced to solve this problem: they constrain executions.
A strategy is a set of states from which a process is allowed to change
round. It can also be seen as a predicate on local states. It captures rules
such as ``wait for at least $F$ messages from the current round'', or ``wait
for these specific messages''. Again, not all strategies lead to executions
with an infinite number of rounds. We then consider valid strategies, which
are strategies that ensure the resulting executions always contain an
infinite number of rounds for each process.

\begin{defi}[Strategy]
    $f \subseteq Q$ is a \textbf{strategy}.
\end{defi}

Strategies as defined above are predicates on states\footnote{One limiting case
is for the strategy to be empty --- the predicate being just $False$. This strategy
is useless, and will be weeded out by the constraint of validity from
Definition~\ref{defVal}.}. This makes them incredibly
expressive, but this expressivity creates difficulty in reasoning
about them. To address this problem, we define families of strategies. Intuitively,
strategies in a same family depend on a specific part
of the state --- for example the messages of the current round.
Equality of these parts of the state defines an equivalence relation;
the strategies of a family are strategies such that if a state $q$ is in the strategy,
then all states in the equivalence class of $q$ are in the strategy.

\begin{defi}[Families of strategies]%
  \label{defFam}
  Let $\approx$ be an equivalence relation on $Q$. The family of strategies defined
  by $\approx$, $\textit{family}(\approx) \triangleq \{ f, \textit{a strategy} \mid \forall
  q_1,q_2 \in Q: q_1 \approx q_2 \implies  (q_1 \in f \iff q_2 \in f)\}$.
\end{defi}

Let's define the executions of a strategy. The intuition is simple: every
change of rounds (an event $next_k$ for $k$ a process) happens only if the
local state of the process is in the strategy. There is a fairness
assumption that ensures that if the local state of some process $k$ is
eventually continuously in the strategy, then it will eventually change
round (have a $next_k$ event)\footnote{This is a weak fairness assumption,
which requires a constant availability of the transition after some point to
ensure it will happen. This is to be contrasted with a strong fairness
assumption, which requires the availability infinitely often after some
point to ensure that the transition will happen.}.
A subtlety hidden in this obvious intuition is that the check for changing
round (whether the local state is in the strategy) doesn't necessarily happen
at each reception; it can happen at any point. This captures an asynchronous
assumption where processes do not decide when they are executed.

\begin{defi}[Executions of a Strategy]%
  \label{defExecStrat}
  Let $f$ be a strategy and $t$ an execution.
  Then $t$ is an \textbf{execution of} $f$ if
  $t$ satisfies:
  \begin{itemize}
      \item (All nexts allowed)
          $\forall i \in \mathbb{N}, \forall p \in \Pi:
          (t[i] = next_p  \implies q_p^t[i] \in f)$,
      \item (Fairness)
          $\forall p \in \Pi:
          \card(\{i \in \mathbb{N} \mid t[i] = next_p\}) < \infty
          \implies \card(\{i \in \mathbb{N} \mid q_p^t[i] \notin f\})
          = \infty$.
  \end{itemize}
  For a delivered predicate $\DEL$, $\execs_f(\DEL)
  \triangleq
  \{t \mid \textit{t an execution of f } \land t \in \execs(\DEL) \}$.
\end{defi}

The first property states that processes only change round (the $next$
transition) when their local state is in the strategy. Fairness ensures that
the only way for a process $p$ to be blocked at a round $r$ is for $p$'s
local state to not be in $f$ an infinite number of times. If the local state
of $p$ is outside of $f$ only a finite number of times, the local state of
$p$ is eventually always in $f$, and the execution must contain another
$next_p$.

Going back to strategies, not all of them are equally valuable.
In general, strategies with executions where processes are blocked forever
at some round are less useful (to implement infinite sequences of
rounds) than strategies without such executions.
The validity of a strategy captures the absence of such an infinite wait.

\begin{defi}[Validity]%
  \label{defVal}
  An execution $t$ is \textbf{valid} if
  $\forall p \in \Pi, \forall N > 0, \exists i \geq N: t[i] = next_p$.

  Let $\DEL$ a delivered predicate and $f$ a strategy.
  Then $f$ is a \textbf{valid strategy} for $\DEL$
  iff $\forall t \in \execs_f(\DEL), t$ is a valid execution.
\end{defi}

Finally, analogous to how we can combine complex predicates through
operations, we can also compose complex strategies through similar operations:

\begin{defi}[Operations on strategies]%
  \label{defOpsStrat}
  Let $f_1, f_2$ be two strategies. The following operations are defined:
  \begin{itemize}
    \item Their \textbf{union} $f_1 \cup f_2$.
    \item Their \textbf{combination} $f_1 \combi f_2 \triangleq
      \{ q_1 \combi q_2 \mid q_1 \in f_1 \land q_2 \in f_2
      \land q_1.round = q_2.round\}$,
      where for $q_1$ and $q_2$ at the same round $r$,
      $q_1 \combi q_2 \triangleq
      \langle r,
          \{ \langle r', k \rangle
          \mid r' > 0 \land k \in q_1(r') \cap q_2(r')\}
      \rangle$.
    \item Their \textbf{succession} $f_1 \leadsto f_2 \triangleq
      f_1 \cup f_2 \cup \{q_1 \leadsto q_2 \mid
      q_1 \in f_1 \land q_2 \in f_2 \}$ where $q_1 \leadsto q_2 \triangleq {}$
      \[ \left\langle q_1.round+q_2.round,
          \left\{ \langle r, k \rangle
          \mid r > 0 \land
          \left(
          \begin{array}{ll}
              k \in q_1(r) & \text{if } r \leq q_1.round\\
              k \in q_2(r-q_1.round) & \text{if } r > q_1.round\\
          \end{array}
          \right) \right\}
      \right\rangle.\]
    \item The \textbf{repetition} of $f_1$, $f_1^{\omega} \triangleq
      \{q_1 \leadsto q_2 \leadsto \dots \leadsto q_k \mid
      k \geq 1 \land q_1,q_2, \dots,q_k \in f_1\}$.
  \end{itemize}
\end{defi}

The intuition behind these operations is analogous to the ones
for predicates:
\begin{itemize}
  \item The union of two strategies is equivalent
    to an OR of the two conditions. For example,
    the union of waiting for at least $n-F$ messages and waiting for all messages but
    the ones from $p_1$ gives a strategy that accepts change
    of round when more than $n-F$ messages are received or when all
    messages except the one from $p_1$ are received.
  \item The combination of two strategies takes all intersections
    of local states in the first strategy and local states in the second.
    For example, combining the strategy that waits at least $n-1$ messages
    for the current round with itself will wait for at least $n-2$ messages.
  \item For succession, the states accepted are those where messages
    up to some round correspond to an accepted state of the first strategy,
    and messages from this round up correspond to an accepted state of the
    second strategy.
  \item Repetition is the next logical step after succession: instead of following
    one strategy with another, the same strategy is repeated again and again.
\end{itemize}

\subsection{Extracting Heard-Of Collections of Executions}

If an execution is valid, then all processes go through an infinite number
of rounds. That is, it captures the execution of a system-level implementation
of rounds where no process blocks forever at some round.
It thus makes sense to speak of the heard-of collection implemented by this execution:
at the end of round $r$ for process $p$, the messages from round $r$ that were
received by $p$ define the heard-of set for $r$ and $p$.

\begin{defi}[Heard-Of Collections Generated by Executions and
                   Heard-Of Predicate Generated by Strategies]%
  \label{defHOExec}
  Let $t$ be a valid execution. Then
   the \textbf{heard-of collection generated by $t$}, $h_t \triangleq {}$
  \[\forall r \in \mathbb{N}^*, \forall p \in \Pi : h_t(r,p) =
  \left\{ k \in \Pi ~\middle|~ \exists i \in \mathbb{N}:
  \left(
  \begin{array}{ll}
      & q_p^t[i].round = r\\
      \land & t[i] = next_p\\
      \land & \langle r,k \rangle \in q_p^t[i].mes\\
  \end{array}
  \right)
  \right\}.\]

  Let $\DEL$ be a delivered predicate, and $f$ be a valid strategy for \DEL\@.
  Then the \textbf{heard-of predicate generated by $f$ on $\DEL$}
  $\triangleq \HO_f(\DEL)
  \triangleq \{ h_t \mid t \in \execs_f(\DEL) \}$.
\end{defi}

As the strategy is valid, this definition is well-founded and the strategy
generates a heard-of predicate from the delivered predicate.

\subsection{Dominating Predicate}

The way to go from a delivered predicate
to a heard-of one is to design a valid strategy for the former that generates
the latter. With different strategies, different heard-of predicates can be generated.
Which one should be considered as the
characterization of the delivered predicate (and of the corresponding
operational model)?
A heard-of predicate generated from a delivered predicate is
an over-approximation of the latter.
To be able to solve as many problems as possible, as many messages as possible
should be received on time. The characterizing heard-of predicate
is thus the smallest such over-approximation of the delivered
predicate, if it exists.
This intuition is formalized by defining a partial order on valid strategies
for a delivered predicate, capturing the implication of
the generated heard-of predicates (the inclusion when considered as sets).
One strategy dominates another if the
heard-of collections it generates are also generated by the other.
Dominating strategies are then the greatest elements for this order. By
definition of domination, all dominating strategies generate the same
dominating heard-of predicate, which characterizes the delivered predicate.

\begin{defi}[Domination Order, Dominating Strategy and Dominating Predicate]%
  \label{defDom}
  Let $\DEL$ be a delivered predicate and let
  $f$ and $f'$ be two valid strategies for \DEL\@.
  Then, $f$~\textbf{dominates} $f'$ for $\DEL$,
  written $f' \prec_{\DEL} f$, if
  $\HO_{f'}(\DEL) \supseteq \HO_f(\DEL)$.
  A greatest element for $\prec_{\DEL}$, if it exists,
  is called a \textbf{dominating strategy} for \DEL\@. Given
  such a strategy $f$, the \textbf{dominating predicate}
  for $\DEL$ is then $\HO_f(\DEL)$.
\end{defi}

\subsection{Standard and Canonical Executions}

In the following sections, we prove properties about dominating strategies, their
invariance by the operations, and the heard-of predicates that they generate.
These proofs rely on reasoning by contradiction:
assume the theorem or lemma is false, and derive a contradiction.
These contradictions take the form of proving that a valid strategy
has an invalid execution; constructing specific executions is therefore
the main technique in these proofs.
This section introduces two patterns for constructing executions: one from a delivered
collection and a strategy, the other from a heard-of collection.

To do so, let's fix $ord$ as some function taking a set and returning
an ordered sequence of its elements --- the specific ordering doesn't matter. This
will be used to ensure the uniqueness of the executions, but the order has no
impact on the results.

\subsubsection{Standard Execution}

The standard execution extracts an execution from a delivered collection. It
follows a loop around a simple pattern: deliver all the messages that
were sent according to the delivered collection,
then change round for all the processes which are allowed to do so by $f$. The intuition
is that it's the simplest way to specify an execution where strategies that look only at
messages from current and previous rounds (as studied latter) always have all the information
available to them. This means that if at one point a first process fails to change round
while using such a strategy, it will be blocked at this round for the rest of the
standard execution.

Given elements $x_1, x_2, x_3, \dots$, the notation $\prod\limits_{i \in \mathbb{N}^*}\! x_i$ is the infinite concatenation $x_1x_2x_3\dots$

\begin{defi}[Standard Execution of a Strategy on a Collection]%
  \label{defStdExec}
  Let $c$ be a delivered collection, and $f$ be a strategy.
  The \textbf{standard execution of $f$ on $c$} is $st(f,c)
  \triangleq \prod\limits_{r \in \mathbb{N}^*} dels_r.changes_r$,
  where
  \begin{itemize}
    \item $dels_1 \triangleq ord(\{ deliver(1,k,j) \mid k \in c(1,j)\})$, the ordered
      sequence of deliveries for messages from round $1$ that will be delivered eventually
      according to $c$.
    \item $changes_1 \triangleq ord(\{ next_j \mid \langle 1,
      \{(1,k) \mid k \in c(1,j)\} \rangle \in f\})$, the ordered sequence of next transitions
      for processes for which the state resulting from the deliveries in $dels_1$ is
      in $f$.
    \item $\forall r > 1: dels_r \triangleq
      ord(\{deliver(round_r^k,k,j) \mid next_k \in changes_{r-1}
      \land k \in c(round_r^k,j)\})$,\\
      with $round_r^k \triangleq 1 + \sum\limits_{r' \in \rint{1,r}} \card(\{next_k\} \cap changes_{r'})$.

      This is the ordered sequence of deliveries of messages
      from processes that changed round during $changes_{r-1}$.
    \item $\forall r > 1: changes_r \triangleq {}$
      {\small\[\left\{
      \begin{array}{ll}
        ord(\{next_j \mid \langle round_r^j, \{(r',k) \mid deliver(r',k,j) \in
        \bigcup\limits_{r'' \in [1,r]} dels_{r''}\} \rangle \in f \}) & \text{if it
          is not empty}\\
        ord(\{stop\}) & \text{otherwise}
      \end{array}
      \right.,\]}
      with $round_r^j \triangleq 1 + \sum\limits_{r' \in \rint{1,r}} \card(\{next_j\} \cap changes_{r'})$.

      This is the ordered sequence of changes of round for processes
      such that their state after the deliveries of $dels_{r}$ is in $f$.
  \end{itemize}
\end{defi}

\noindent
The main property of a standard execution of $f$ on $c$ is that it
is both an execution of $f$ and an execution of $c$.

\begin{lem}[Correctness of Standard Execution]%
  \label{stCorrect}
  Let $c$ be a delivered collection and $f$ be a strategy.
  Then $st(f,c) \in \execs_f(c)$.
\end{lem}

\begin{proof}[Proof idea]
First, the proof shows that $st(f,c)$ is indeed an execution by verifying
the three properties of Definition~\ref{defExec}. Then, it shows that
$st(f,c)$ is an execution of the delivered collection $c$
(Definition~\ref{defExecColl}: the delivered messages in $st(f,c)$ are
exactly those from $c$). Lastly the verification of the two conditions of
Definition~\ref{defExecStrat} ensures that $st(f,c)$ is an execution of
strategy $f$. This proof is trivial but verbose, and is presented in Appendix~\ref{app:proofStCorrect}.
\end{proof}

\subsubsection{Canonical Execution}

Whereas the standard execution captures a straightforward way to create an
execution from a strategy and a delivered collection, this new construction
starts with a heard-of collection, and creates a valid execution generating
it from the total delivered collection (the one with all the receptions).
The link between the two is that when we want to prove that a valid strategy $f$ implements
a heard-of collection $ho$ (useful for showing dominance of strategies), we show by
contradiction that the canonical execution of $ho$ is an execution of $f$. Since
the proof by contradiction assumes that one of the $next_j$ transitions in the canonical
execution is not allowed by $f$, we can usually find a delivered collection
where the messages delivered at this point are the only ones that will ever be delivered to $j$,
and so show that the standard execution of this delivered collection is not valid,
which contradicts the validity of $f$.

This canonical execution follows a simple pattern: at each round,
deliver all the messages from the heard-of sets of
this round, and also all the messages undelivered from the previous round (the ones
that were not in the heard-of sets of the previous round).

\begin{defi}[Canonical Execution of a Heard-Of Collection]%
  \label{defCanExec}
  Let \textit{ho} be a heard-of collection.
  The \textbf{canonical execution of} \textit{ho} is $can(ho)
  \triangleq \prod\limits_{r \in \mathbb{N}^*} dels_r.changes_r$,
  where
  \begin{itemize}
    \item $dels_1 \triangleq ord(\{deliver(1,k,j) \mid k \in ho(1,j)\})$, the ordered sequence
      of deliveries that happen at round $1$ in $h$.
    \item $\forall r > 1: dels_r \triangleq ord(\{deliver(r,k,j) \mid k \in ho(r,j)\} \cup
      \{deliver(r-1,k,j) \mid k \in \Pi \setminus ho(r-1,j)\})$, the ordered sequence of
      deliveries that happen at round $r$ in $h$.
    \item $\forall r > 0: changes_r \triangleq ord(\{next_k \mid k \in \Pi\})$, the
      ordered sequence of next transitions, one for each process.
  \end{itemize}
\end{defi}

\noindent
This canonical execution is an execution of any delivered predicate containing
$c_{total}$, the collection where every message is delivered.
Having this collection in a delivered predicate
ensures that although faults might happen, they are not forced to do so.

\begin{lem}[Canonical Execution is an Execution of Total Collection]%
  \label{trivial}
  Let \textit{ho} be a heard-of collection.
  The canonical execution $can(ho)$ of \textit{ho} is
  an execution of $c_{total}$.
\end{lem}

\begin{proof}
    First, $can(ho)$ is an execution by Definition~\ref{defExec}
    since it satisfies the three conditions:
    \begin{itemize}
      \item Delivered only once:
        Every sent message is delivered either during the
        round it was sent or during the next one, and thus
        delivered only once.
      \item Delivered after sending:
        Every message from round $r$ is delivered after
        either $r-1$ or $r$ $next_p$ transitions for
        the sender $p$, which means at round $r$ or $r+1$. Hence
        the message is delivered after being sent.
      \item Once stopped, forever stopped:
        No process stops, so the last condition for executions
        is trivially satisfied.
    \end{itemize}

    \noindent
    Furthermore, for each process $p$ and round $r$, all
    the messages from $p$ at round $r$ are delivered in $can(ho)$,
    either at round $r$ or at round $r+1$.
    Since the total collection is the collection where every message is delivered
    by Definition~\ref{total}, this entails
    that $can(ho)$ is an execution of the total delivered
    collection by Definition~\ref{defExecColl}, and thus an execution of \DEL\@.
\end{proof}

Lastly, the whole point of the canonical execution of $ho$ is that it
generates $ho$.

\begin{lem}[Canonical Execution Generates its Heard-Of Collection]%
  \label{canHO}
  Let \textit{ho} be a heard-of collection. Then $h_{can(ho)} = ho$.
\end{lem}

\begin{proof}

Let  $r > 0$ and  $j \in \Pi$. Let's prove that $h_{can(ho)}(r,j) = ho(r,j)$.

\begin{itemize}
\item ($\subseteq$) Let $p \in h_{can(ho)}(r,j)$, let's prove that $p \in ho(r,j)$.
Since $p \in  h_{can(ho)}(r,j)$, Definition~\ref{defHOExec} gives: $\exists i \in \mathbb{N}: q_j^{can(ho)}[i].round = r  \land can(ho)[i] = next_j \land \langle r,p \rangle \in q_j^{can(ho)}[i].mes$.
	\begin{itemize}
	\item Since $q_j^{can(ho)}[i].round = r$, Definition~\ref{defLocalState} gives $\card(\{l <i : can(ho)[l]=next_j\})=r-1$.
	\item Definition~\ref{defCanExec} implies that $next_j$ appears only in the $changes_s$ sequences, and only once in each $changes_s$. So, since $can(ho)[i] = next_j$ and $\card(\{l <i : can(ho)[l]=next_j\})=r-1$, then $can(ho)[i]$ is in $changes_r$ and so is between $dels_r$ and $dels_{r+1}$.
	\item Since $\langle r,p \rangle \in q_j^{can(ho)}[i].mes$ and  $can(ho)[i]$ is between $dels_r$  and $dels_{r+1}$, Definition~\ref{defCanExec} implies that  $p \in ho(r,j)$.
	\end{itemize}
\item ($\supseteq$) Let $p \in ho(r,j)$, let's prove that $p \in h_{can(ho)}(r,j)$, i.e., (Definition~\ref{defHOExec}) $\exists i \in \mathbb{N}: q_j^{can(ho)}[i].round = r  \land can(ho)[i] = next_j \land \langle r,p \rangle \in q_j^{can(ho)}[i].mes$.
  Let $i$ the index of $changes_r$, such that $can(ho)[i] = next_j$. Definition~\ref{defCanExec} ensures such an $i$ exists since in each $changes_s$, there is a $next_k$ for each $k \in \Pi$.
\begin{itemize}
\item Definition~\ref{defCanExec} ensures that there are $r-1$ $next_j$ transitions in $can(ho)$ before index $i$, so $q_j^{can(ho)}[i].round = r$.
\item By definition of $i$, $can(ho)[i] = next_j$.
\item Definition~\ref{defCanExec} ensures that $dels_r$ is before $changes_r$ in $can(ho)$. Since $p \in ho(r,j)$, so $deliver(r, p, j)$ appears in $can(ho)$ before index $i$ and so $\langle r,p \rangle \in q_j^{can(ho)}[i].mes$.
\qedhere
\end{itemize}
\end{itemize}
\end{proof}

\subsection{A Complete Example: At Most \texorpdfstring{$F$}{F} Crashes}%
\label{subsec:chap2formalexample}

Let's look at a concrete example to get a better grasp at how all these concepts
work together, and the kind of results they allow us to prove.
We consider $\DEL^{crash}_F$ from Definition~\ref{delCrash},
the message-passing model with
asynchronous and reliable communication, and at most $F$ permanent
crashes. We now give a dominating strategy for this predicate, as well as compute
its heard-of predicate.
The folklore strategy for this model is to wait for at least $n-F$ messages
before allowing the change of round.

\begin{defi}[waiting for $n-F$ messages]%
  \label{deffn-f}
    The strategy to wait for $n-F$ messages is
    $f_{n-F} \triangleq
    \{ q \in Q \mid
    \card(\{k \in \Pi \mid \langle q.round, k \rangle \in q.mes\}) \geq n-F \}$.
\end{defi}

To see why this strategy is used in the literature, remark that
at least $n-F$ messages must be delivered to each process at each round.
Thus, waiting for that many messages ensures that no process is ever
blocked. However, waiting for more will result in processes blocking forever
if $F$ crashes occur. Rephrased with the concepts introduced above,
$f_{n-F}$ is a valid and dominating strategy for $\DEL^{crash}_F$.

\begin{lem}[Validity of $f_{n-F}$]%
    \label{valid}
    $f_{n-F}$ is valid for $\DEL^{crash}_F$.
\end{lem}

\begin{proof}
    We proceed by contradiction. \textbf{Assume} $f_{n-F}$ is invalid
    for $\DEL^{crash}_F$. By Definition~\ref{valid},
    there exists an invalid
    $t \in \execs_{f_{n-F}}(\DEL^{crash}_F)$.
    By Definition~\ref{valid} of validity, $\exists p \in \Pi, \exists N, \forall i
    \geq N: t[i] \neq next_p$: there is a smallest round $r$
    such that some process $j$ stays blocked
    at $r$ forever in $t$. Because $t$ is an execution of $f$,
    Definition~\ref{defExecStrat} entails that infinitely many local states of
    $j$ must be not in $f_{n-F}$; if it was not the case, the fairness condition
    would force the execution to contain another $next_j$.
    Let also $c_t$ be a delivered collection of $\DEL^{crash}_F$
    such that $t \in \execs(c)$.
    We know by Definition~\ref{delCrash} of $\DEL^{crash}_F$ that
    $\card(c_t(r,j)) \geq n-F$. The minimality of $r$ and
    the fact that $t \in \execs(c)$ ensure by Definition~\ref{defExecColl}
    that all the messages in this delivered set are delivered at some point
    in $t$. By Definition~\ref{deffn-f} of $f_{n-F}$, the local state of
    $j$ is then in $f_{n-F}$ from this point on.
    By the fairness constraint of Definition~\ref{defExecStrat},
    this \textbf{contradicts} the fact that there is never another
    $next_j$ in the suffix of $t$.
    We conclude that $f_{n-F}$ is valid for $\DEL^{crash}_F$.
\end{proof}

The next step is to prove that this strategy is dominating the predicate. We
first need to compute the heard-of predicate generated by $f_{n-F}$. This
heard-of predicate was given by Charron-Bost and
Schiper~\cite{CharronBostHO} as a characterization of the asynchronous model
with reliable communication and at most $F$ crashes, without a formal proof. The intuition behind it
is that even in the absence of crashes, we can make all the processes change
round by delivering any set of at least $n-F$ messages to them.

\begin{thm}[Heard-Of Characterization of $f_{n-F}$]%
    \label{n-fCharac}~\\
    Let $HO_F$ be the heard-of predicate defined as $\{ h, \textit{a heard-of collection} \mid \forall r > 0, \forall p \in \Pi: \card(ho(r,p)) \geq n-F\}$.
    Then $\HO_{f_{n-F}}(\DEL^{crash}_F) = HO_F$.
\end{thm}

\begin{proof}\hfill
    \begin{itemize}
      \item $(\subseteq)$.
        Let $\textit{ho} \in \HO_{f_{n-F}}(\DEL^{crash}_F)$
        and $t \in \execs_{f_{n-F}}(\DEL^{crash}_F)$ an execution of $f_{n-F}$
        generating \textit{ho}.
        By Definition~\ref{defExecStrat} of the executions of $f_{n-F}$, processes change
        round (a $next_k$ event) only when their local state is in
        $f_{n-F}$. By Definition~\ref{deffn-f}, This means
        that this local state $q$ satisfies
        $\card(\{k \in \Pi \mid \langle q.round, k \rangle \in q.mes\}) \geq n-F \}$:
        the process received at least $n-F$ messages tagged with the current
        value of its round counter. This in turns implies by Definition~\ref{defHOExec}
        of the heard-of collection of an execution that
        $\forall r \in \mathbb{N}^*, \forall j \in \Pi : \card(ho(r,j)) \geq n-F$.
      \item $(\supseteq)$.
        Let \textit{ho} a heard-of collection over $\Pi$ such that
        $\forall r \in \mathbb{N}, \forall j \in \Pi: \card(\textit{ho}(r,j)) \geq n-F$.
        Let $t$ be the canonical execution of \textit{ho}; since $\DEL^{crash}_F$
        contains the total collection, $t$ is an execution of $\DEL^{crash}_F$
        by Lemma~\ref{trivial}.
        To prove that $t$ is also an execution of $f_{n-F}$, we proceed by contradiction.
        \textbf{Assume} it is not an execution of $f_{n-F}$.
        By Definition~\ref{defExecStrat}, the problem stems either from breaking fairness
        or from some $next_p$ for some $p$ at a point where the local state of $p$ is not
        in $f_{n-F}$. Since by Definition~\ref{defCanExec} of a canonical execution,
        $\forall p \in \Pi: next_p$ appears an infinite number of times,
        the only possibility left is the second one: there is some $p \in \Pi$ and
        some $next_p$ transition in $t$ that happens while the local state of $p$
        is not in $f_{n-F}$.
        Let $r$ be the smallest round where this happens, and $j$ the process
        to which it happens.
        By Definition~\ref{defCanExec} of a canonical execution,
        $j$ received all the messages from $ho(r,j)$ in $t$
        before the problematic $next_j$. By hypothesis, $\card(ho(r,j)) \geq n-F$.
        By Definition~\ref{deffn-f} of $f_{n-F}$, the
        local state of $j$ is in $f_{n-F}$ from this
        point on.
        By the fairness constraint of Definition~\ref{defExecStrat},
        this \textbf{contradicts} the fact that $j$ cannot
        change round at this point in $t$. Hence $t$ is an execution of $f_{n-F}$.
    \end{itemize}
    We conclude that $\textit{ho} \in \HO_{f_{n-F}}(\DEL^{crash}_F)$.
\end{proof}

Finally, we want to vindicate the folklore intuition about this strategy:
that it is optimal in some sense. Intuitively, waiting for more than $n-F$
messages per round means risking waiting forever, and waiting for less
is wasteful. Our domination order captures this concept of optimality:
we show that $f_{n-F}$ is indeed a dominating strategy for $\DEL^{crash}_F$.
Therefore, $\HO_{f_{n-F}}(\DEL^{crash}_F)$ is the
dominating predicate for $\DEL^{crash}_F$.

\begin{thm}[$f_{n-F}$ Dominates $\DEL^{crash}_F$]%
    \label{fnfDom}
    $f_{n-F}$ dominates $\DEL^{crash}_F$.
\end{thm}

\begin{proof}
    Let $f$ be a valid strategy for $\DEL^{crash}_F$; the theorem
    follows by Definition~\ref{defDom} from proving
    that $f \prec_{\DEL^{crash}_F} f_{n-F}$ --- that is
    $\HO_{f_{n-F}}(\DEL^{crash}_F) \subseteq
    \HO_f(\DEL^{crash}_F)$.

    Let $\textit{ho} \in \HO_{f_{n-F}}(\DEL^{crash}_F)$, and
    let $t$ be the canonical execution of \textit{ho}. Since $\DEL^{crash}_F$
    contains the total collection, $t$ is an execution of $\DEL^{crash}_F$
    by Lemma~\ref{trivial}. We need to prove that it is also an execution
    of $f$ to conclude by Lemma~\ref{canHO} that $f$ generates $ho$, and thus
    that the inequality above and the theorem hold.
    We do so by contradiction.
    \textbf{Assume} $t$ is not an execution of $f$.
    By Definition~\ref{defExecStrat}, it is either because the fairness
    condition is broken or because some $next_p$ for some process $p$ happens
    when the local state of $p$ is not in $f$. Since Definition~\ref{defCanExec}
    of canonical executions implies that $t$ contains an infinite number of $next_p$
    for every process $p \in \Pi$,
    the problem must come from some $next_j$ done by a process $j$ when
    $j$'s local state is not in $f$. Let $r$ be the first round where this
    happens.
    At the point of the forbidden $next_j$, by Definition~\ref{defCanExec} of a
    canonical execution, $j$ has received
    all the messages from previous rounds, and all the messages from
    $\textit{ho}(r,j)$. Then $ho \in \HO_{f_{n-F}}(\DEL^{crash}_F)$
    implies that $ho \in HO_F$ by Theorem~\ref{n-fCharac}. It then follows
    from the definition of $HO_F$ that
    $\textit{ho}(r,j)$ contains at least $n-F$ processes.

We define $c_{\textit{block}}$ such that:
\[\forall r' > 0, \forall k \in \Pi:
    c_{\textit{block}}(r',k) \triangleq
    \left\{
    \begin{array}{ll}
        \Pi & \textit{if } r' < r\\
        ho(r,j) & \textit{otherwise}\\
    \end{array}
    \right.\]
This corresponds to the collection where all the processes from which $j$ did not receive a message at the problematic $next_j$ in $t$ are correct up to round $r$ and then stop sending messages.
    $c_{\textit{block}}$ is a delivered collection of $\DEL^{crash}_F$
    by Definition~\ref{delCrash}:
    processes that stop sending messages never do again,
    and at most $F$ processes do so because $\textit{ho}(r,j)$
    contains at least $n-F$ processes by the reasoning above.

    Let $t' = st(f,c_{\textit{block}})$ be the standard execution of
    $f$ on $c_{\textit{block}}$.
    Lemma~\ref{stCorrect} entails than $t'$ is an execution of $f$
    on $c_{\textit{block}}$.
    Since $r$ is the smallest round in $t$ with a wrong $next_j$, for
    all rounds $< r$ the local state of $j$ is enough for $f$ to allow
    the change of round.
    By Definition~\ref{defStdExec} of standard executions,
    all $changes_k$ for $k < r$ contain a next transition for all
    processes in $t'$. By the same definition, all $dels_k$ for
    $k \leq r$ of $t'$ contain the same deliveries for each process as
    the deliveries for $j$ in the $dels_k$ of $t$.
    Hence, in $t'$, all the processes reach round $r$, all get the same state
    as $j$ in $t$ at round $r$, and thus they all block at this round, which means
    the suffix of $t'$ is made of $stop$ only. $t'$ is invalid,
    and so is $f$.
    This \textbf{contradicts} the hypothesis that $f$ is valid; we conclude that
    $ho \in \HO_f(\DEL^{crash}_F)$.
    Therefore, $f_{n-F}$ dominates $f$ by Definition~\ref{defDom}, where $f$ is
    any valid strategy for $\DEL^{crash}_F$, which means that $f_{n-F}$ dominates
    $\DEL^{crash}_F$ by Definition~\ref{defDom}.
\end{proof}

Waiting for $n-f$ messages thus gives us the best heard-of predicate that
can be implemented on $\DEL^{crash}_F$. This means that
there is no point in remembering messages from past
rounds, and messages from future rounds are not used either.
Intuitively, messages from past rounds are of no use in detecting
crashes in the current round. As for messages from
future rounds, they actually serve to detect that a process has not crashed when
sending its messages from the current round. For this to actually change the heard-of
predicate, it would require that some heard-of collection be impossible to
generate when using this future information. This is not the case, as
there is always an execution where no message from future
rounds are delivered early (the canonical execution).

\section{Oblivious Strategies:  Looking Only at the Current Round}%
\label{sec:chap2obliv}

Because of the generality of strategies, considering them all brings many issues
in proving domination. Yet there exist interesting
classes of strategies on which results can be derived.
Our first such class is the class of oblivious strategies: they depend only
on the received messages from the current round. For example, $f_{n-F}$ is
an oblivious strategy, as it counts messages from the current round.
Despite their apparent simplicity, some oblivious strategies dominate non-trivial
delivered predicates, as in the case of $f_{n-F}$ and $\DEL^{crash}_F$.

In this section, we define oblivious strategies and give a necessary and
sufficient condition for an oblivious strategy to be valid
(Section~\ref{sec:oblivious:def}). This yields useful results on the
composition of oblivious strategies
(Section~\ref{sec:oblivious:composition}), and enables to compute the
heard-of predicate of an oblivious strategy applied to a delivered predicate
(Section~\ref{sec:oblivious:heardof}). Finally, we give a sufficient
condition for oblivious domination, and show that this condition is preserved by composition
(Section~\ref{sec:oblivious:domination}).

\subsection{Definition and Expressiveness}%
\label{sec:oblivious:def}

Oblivious strategies are a family of strategies in the sense of
Definition~\ref{defFam} --- where the equivalence relation between the local
states compare only the messages received for the current round.

\begin{defi}[Oblivious Strategies and $\Nexts_f$]%
  \label{defObliv}
    Let $\obliv$ be the function such that $\forall q \in Q: \obliv(q) \triangleq
    \{k \in \Pi \mid \langle q.round, k \rangle \in q.mes\}$.
    Let $\approx_{\obliv}$ the equivalence relation defined by
    $q_1 \approx_{\obliv} q_2$ iff $\obliv(q_1) = \obliv(q_2)$.
    The family of \textbf{oblivious strategies} is $\textit{family}(\approx_{\obliv})$.

    For $f$ an oblivious strategy, let $\Nexts_{f} \triangleq
    \{\obliv(q) \mid q \in f \}$.
    It uniquely defines $f$.
\end{defi}

An oblivious strategy reduces to a collection of sets, the sets
of processes from which receiving a message in the current round is enough
to change round. The strategy allows the change of round if,
and only if, the processes heard at the current round form a set in this
collection.
This provides a simple necessary condition on such a strategy $f$ to be valid:
its $\Nexts_f$ set
must contain all the delivered sets from the corresponding delivered predicate.
If it does not, an execution would exist where the messages received
at some round $r$ by some process $p$ are exactly this delivered set, which
would block forever the process $p$ and make the strategy invalid.
This necessary condition also proves sufficient.

\begin{lem}[Necessary and Sufficient Condition for Validity of an Oblivious Strategy]%
    \label{oblivValid}
    Let $\DEL$ be a delivered predicate and $f$ be an oblivious strategy.
    Then $f$ is valid for $\DEL \iff
    f \supseteq \{q \mid \exists c \in \DEL, \exists p \in \Pi,
    \exists r  > 0: \obliv(q) = c(r,p) \}$.
\end{lem}

\begin{proof}\hfill
  \begin{itemize}
    \item $(\Rightarrow)$ Let $f$ be valid for \DEL\@.
      We show by contradiction that $f$ contains all local states $q$ such that
      $\exists c \in \DEL, \exists p \in \Pi, \exists r  > 0:
      \obliv(q) = c(r,p)$.
      \textbf{Assume} there is some $q_{block}$ for which it is not the case: then
      $\exists c \in \DEL, r > 0$ and $j \in \Pi$ such that
      $\obliv(q_{block}) = c(r,j)$ and $q_{block} \notin f$. By
      Definition~\ref{defObliv}, this means that for every $q$ such that
      $\obliv(q) = \obliv(q_{block}) = c(r,j)$, we have $q \notin f$.
      Let $t = st(f,c)$ be the standard execution of $f$ on $c$.
      This is an execution of $f$ on $c$ by Lemma~\ref{stCorrect}.
      The sought contradiction is reached by proving that $t$ is invalid.
      To do so, we split according to two cases: the first case is where
      there is a blocking process before round $r$, and the second case is
      where there is no blocking process before round $r$. This last case
      then uses the hypothesis on $c(r,j)$ to show that all the processes block at
      $r$.
      \begin{itemize}
          \item During one of the first $r-1$ iterations of
              $t$, there is some process which cannot change
              round. Let $r'$ be the smallest iteration where it happens,
              and $k$ be a process unable to change round at the $r'$-th
              iteration.
              By minimality of $r'$, all the processes arrive at round $r'$ in
              $t$; by Definition~\ref{defStdExec} of the standard execution,
              all messages for $k$ from round $r'$ are delivered before
              the $change_{r'}$ part of the iteration.
              Let $q$ be the local state of $k$ at the
              start of $change_{r'}$ in the $r'$-th iteration, and
              let $q'$ be any local state of $k$ afterward.
              The above tells us that as long as $q'.round = q.round$,
              we have $\obliv(q) = \obliv( q')$ and thus $q' \notin f$.
              Therefore, $k$ can never change round while at round $r'$.
              We conclude that $t$ is invalid by Definition~\ref{valid}.

          \item For the first $r-1$ iterations, all the processes change round.
              Thus everyone arrives at round $r$ in the $r-1$-th iteration.
              By Definition~\ref{defStdExec} of the standard execution,
              all messages from the round are delivered
              before the $change_{r}$ part of the $r$-th iteration.
              Thus $j$ is in a local state $q$
              at the $change_{r}$ part of the $r$-th iteration
              such that $\obliv(q) = c(r,j) =
              \obliv(q_{block})$. By hypothesis, this means $q \notin f$ and
              thus that $j$ cannot change round.
              Let $q'$ be any local state of $j$ afterward.
              The above tells us that as long as $q'.round = r$,
              we have $\obliv(q) = \obliv(q') = c(r,j)$
              and thus $q' \notin f$. Therefore, $j$ can never
              change round while at round $r$.
              Here too, $t$ is invalid by Definition~\ref{valid}.
      \end{itemize}
      Either way, we reach a \textbf{contradiction} with the validity of $f$.
      Therefore $f \supseteq \{q \mid \exists c \in \DEL, \exists p \in \Pi,
      \exists r  > 0: \obliv(q) = c(r,p) \}$.

    \item $(\Leftarrow)$ Let $\DEL$ and $f$ be such that
      $f \supseteq \{q \mid \exists c \in \DEL, \exists p \in \Pi,
      \exists r  > 0: \obliv(q) = c(r,p) \}$.
      We show by contradiction that $f$ is valid.
      \textbf{Assume} the contrary: there is some
      $t \in \execs_f(\DEL)$ which is invalid. By
      Definition~\ref{valid} of validity, there
      are some processes blocked at a round forever in $t$. Let $r$ be the
      smallest such round, and $j$ be a process blocked at round $r$ in $t$.
      By minimality of $r$, all the processes arrive at round $r$.
      By Definition~\ref{defExecColl} of an execution of \DEL, there is a
      $c \in \DEL$ such that $t$ is an execution of $c$.
      This means by Definition~\ref{defExecColl} of an execution of a collection
      that all the messages from $c(r,j)$ are eventually delivered.
      From this point on, every local state $q$ of $j$ satisfies $\obliv(q) = c(r,j)$;
      thus we have $q \in f$ by hypothesis on $f$.
      Then the fairness condition of executions of $f$
      from Definition~\ref{defExecStrat} imposes that $j$ does change round at some point.
      We conclude that $j$ is not blocked at round $r$ in $t$,
      which \textbf{contradicts} the hypothesis that it is blocked forever at round $r$
      in~$t$.
      \qedhere
  \end{itemize}
\end{proof}

\noindent
When taking the oblivious strategy satisfying exactly this condition
for validity, it results in a strategy dominating the other oblivious ones.
It follows from the fact that
this strategy waits for the minimum sets required to be valid, hence the
name of minimal oblivious strategy.

\begin{defi}[Minimal Oblivious Strategy]%
    \label{defMinObliv}
    Let $\DEL$ be a delivered predicate. The \textbf{minimal oblivious strategy}
    for $\DEL$ is $f_{\min} \triangleq
    \{q \mid \exists c \in \DEL, \exists p \in \Pi,
    \exists r  > 0: \obliv(q) = c(r,p) \}$.
\end{defi}

\begin{lem}[Domination of Minimal Oblivious Strategy]%
    \label{domMinObliv}
    Let $\DEL$ be a delivered predicate and $f_{\min}$ be its minimal oblivious strategy.
    Then $f_{\min}$ is a dominating oblivious strategy for \DEL\@.
\end{lem}

\begin{proof}
    First, $f_{\min}$ is valid for $\DEL$ by application of
    Lemma~\ref{oblivValid}.
    Next, we take another oblivious strategy $f$, which is valid
    for \DEL\@. Lemma~\ref{oblivValid} now gives us that
    $f_{\min} \subseteq f$. When
    $f_{\min}$ allows a change of round, so does $f$.
    This entails by Definition~\ref{defExecStrat}
    that all executions of $f_{\min}$ on $\DEL$ are also executions of $f$ on $\DEL$,
    and thus by Definition~\ref{defHOExec}
    that $\HO_{f_{\min}}(\DEL) \subseteq \HO_f(\DEL)$.
    We conclude from Definition~\ref{defDom}
    that $f_{\min}$ dominates any valid oblivious strategy for \DEL\@.
\end{proof}

These theorems guarantee that every delivered predicate has a
strategy dominating the oblivious ones, by giving a means to build it.
Of course, a formal definition is not the same as a constructive definition,
which motivates the study of minimal strategies through the operations,
and their relations to the operations on the corresponding predicates.

\subsection{Composition of Oblivious Strategies}%
\label{sec:oblivious:composition}

One fundamental property of minimal oblivious strategies is their nice
behaviour under the proposed operations (union, combination, succession and repetition).
That is, they give minimal oblivious strategies of resulting delivered predicates.
Although this holds for all operations, succession and repetition are not useful
here, as the succession of two minimal oblivious strategies is equal to their union,
and the repetition of a minimal oblivious strategy is equal to the strategy itself.
The first operation to study is therefore union.
The minimal oblivious strategy of $\DEL_1 \cup \DEL_2$ and $\DEL_1 \leadsto \DEL_2$
is the same, as shown in the next theorem, and thus it's
the union of the minimal oblivious strategies of $\DEL_1$ and $\DEL_2$.

\begin{thm}[Minimal Oblivious Strategy for Union and Succession]%
    \label{unionObliv}
    Let $\DEL_1, \DEL_2$ be two delivered predicates, $f_1$ and $f_2$ the minimal
    oblivious strategies for, respectively, $\DEL_1$ and $\DEL_2$.
    Then $f_1 \cup f_2$ is the minimal oblivious strategy for
    $\DEL_1 \cup \DEL_2$ and $\DEL_1 \leadsto \DEL_2$.
\end{thm}

\begin{proof}[Proof idea]
  Structurally, every proof in this subsection amounts to showing
  equality between the strategies resulting from the operations
  and the minimal oblivious strategy for the delivered predicate.
  For a union, the messages that can be received at each round are
  the messages that can be received at each round in the first
  predicate or in the second. This is also true for succession.
  Given that $f_1$ and $f_2$ are the minimal oblivious strategies
  of $\DEL_1$ and $\DEL_2$,
  they accept exactly the states where the messages received from the
  current round are in a delivered set of $\DEL_1$ or a delivered set
  of $\DEL_2$. Thus $f_1 \cup f_2$ is the minimal oblivious strategy for
  $\DEL_1 \cup \DEL_2$ and $\DEL_1 \leadsto \DEL_2$.
\end{proof}

\begin{proof}
    We first show that the minimal oblivious strategies of
    $\DEL_1 \cup \DEL_2$ and $\DEL_1 \leadsto \DEL_2$ are equal.
    By Definition~\ref{defOpsStrat}, we need to prove that
    $\{q \mid \exists c \in \DEL_1 \cup \DEL_2, \exists p \in \Pi,
    \exists r > 0: \obliv(q) = c(r,p) \} =
    \{q \mid \exists c \in \DEL_1 \leadsto \DEL_2, \exists p \in \Pi,
    \exists r > 0: \obliv(q) = c(r,p) \}$.

    \begin{itemize}
        \item $(\subseteq)$ Let $q$ be such that $\exists c \in \DEL_1
          \cup \DEL_2, \exists p \in \Pi, \exists r > 0: \obliv(q)
          = c(r,p)$.
          \begin{itemize}
            \item If $c \in \DEL_1$, then we take $c_2 \in \DEL_2$,
              and take $c' = c[1,r].c_2$. Since by Definition~\ref{defOpsPred}
              $c' \in c \leadsto c_2$, we have $c' \in \DEL_1 \leadsto \DEL_2$.
              By definition of $c'$, $c'(r,p) = c(r,p)$.
              We thus have $c', p$ and $r$ showing that
              $\obliv(q) = c'(r,p)$, and thus $q$
              is in the set on the right.
            \item If $c \in \DEL_2$, then
              $c \in \DEL_1 \leadsto \DEL_2$ by Definition~\ref{defOpsPred}.
              We thus have $c, p$ and $r$ showing that $\obliv(q) = c(r,p)$,
              and thus $q$ is in the set on the right.
          \end{itemize}
        \item $(\supseteq)$ Let $q$ be such that $\exists c \in \DEL_1
          \leadsto \DEL_2, \exists p \in \Pi, \exists r > 0: \obliv(q)
          = c(r,p)$.
          \begin{itemize}
            \item If $c \in \DEL_2$, then $c \in \DEL_1 \cup
              \DEL_2$ by Definition~\ref{defOpsPred}.
              We thus have $c, p$ and $r$ showing that $\obliv(q) = c(r,p)$,
              and thus $q$ is in the set on the left.
            \item If $c \notin \DEL_2$, there exist
              $c_1 \in \DEL_1, c_2 \in \DEL_2$ and
              $r' > 0$ such that $c = c_1[1,r'].c_2$
              by Definition~\ref{defOpsPred}.
              \begin{itemize}
                \item If $r \leq r'$, then by definition of
                  $c$, we have $c(r,p) = c_1(r,p)$.
                  We thus have $c_1, p$ and $r$ showing that
                  $\obliv(q) = c_1(r,p)$, and thus $q$
                  is in the set on the left.
                \item If $r > r'$, then $c(r,p) = c_2(r-r',p)$.
                  We thus have $c_2, p$ and $(r-r')$ showing that
                  $\obliv(q) = c_2(r-r',p)$, and thus $q$
                  is in the set on the left.
              \end{itemize}
          \end{itemize}
    \end{itemize}

    \noindent
    We now show that $f_1 \cup f_2 =
    \{q \mid \exists c \in \DEL_1 \cup \DEL_2, \exists p \in \Pi,
    \exists r > 0: \obliv(q) = c(r,p) \}$, which allows us to conclude
    by Definition~\ref{defMinObliv} that $f_1 \cup f_2$ is the minimal
    oblivious strategy for $\DEL_1 \cup \DEL_2$.
    \begin{itemize}
        \item Let $q \in f_1 \cup f_2$. We fix $q \in f_1$
            (the case $q \in f_2$ is symmetric).
            Then because $f_1$ is the minimal oblivious strategy of $\DEL_1$,
            by application of Lemma~\ref{oblivValid},
            $\exists c_1 \in \DEL_1,\exists p \in \Pi, \exists r > 0$
            such that $c_1(r,p)= \obliv(q)$. Also,
            $c_1 \in \DEL_1 \subseteq \DEL_1 \cup \DEL_2$ by
            Definition~\ref{defOpsPred}.
            We thus have $c_1, p$ and $r$ showing that $q$
            is in the minimal oblivious strategy for $\DEL_1 \cup \DEL_2$.
        \item Let $q$ be such that $\exists c \in \DEL_1 \cup \DEL_2,
            \exists p \in \Pi, \exists r > 0: c(r,p)= \obliv(q)$.
            By Definition~\ref{defOpsPred} of operations on strategies, and
            specifically union, $c$ must be in $\DEL_1$ or
            $c$ must be in $\DEL_2$; we fix $c \in \DEL_1$
            (the case $\DEL_2$ is symmetric).
            Then Definition~\ref{defMinObliv} gives us that
            $q$ is in the minimal oblivious strategy of $\DEL_1$, that is $f_1$.
            We conclude that $q \in f_1 \cup f_2$.
            \qedhere
    \end{itemize}
\end{proof}

\noindent
For the same reason that succession is indistinguishable from union,
repetition is indistinguishable from the original predicate: the delivered
sets are the same, because every collection of the repetition is built from
prefixes of collections of the original predicate.
Thus, the minimal oblivious strategy for a repetition is
the same strategy as the minimal oblivious strategy of the original
predicate.

\begin{thm}[Minimal Oblivious Strategy for Repetition]%
    \label{repetObliv}
    Let $\DEL$ be a delivered predicate, and
    $f$ be its minimal oblivious strategy.
    Then $f$ is the minimal oblivious strategy for
    $\DEL^{\omega}$.
\end{thm}

\begin{proof}
    We show that $f =
    \{q \mid \exists c \in \DEL^{\omega}, \exists p \in \Pi,
    \exists r > 0: \obliv(q) = c(r,p) \}$, which allows us to conclude
    by Definition~\ref{defMinObliv} that $f$ is the minimal oblivious strategy for
    $\DEL^{\omega}$.
    \begin{itemize}
        \item $(\subseteq)$ Let $q \in f$. By minimality of $f$
            for $\DEL$,
            $\exists c \in \DEL, \exists p \in \Pi, \exists r > 0:
            \obliv(q) = c(r,p)$.
            We take $c' \in \DEL^{\omega}$ such that
            $c_1 = c$ and $r_2 = r$; the other $c_i$ and
            $r_i$ don't matter for the proof.
            By Definition~\ref{defOpsPred} of operations on predicates,
            and specifically repetition, we get
            $c'(r,p) = c(r,p) = \obliv(q)$.
            We have $c', p$ and $r$ showing that
            $q$ is in the minimal oblivious strategy of
            $\DEL^{\omega}$.
        \item $(\supseteq)$ Let $q$ be such that
            $\exists c \in \DEL^{\omega}, \exists p \in \Pi, \exists r > 0:
            \obliv(q) = c(r,p)$.
            By Definition~\ref{defOpsPred} of operations on predicates,
            and specifically repetition, there are
            $c_i \in \DEL$ and $0< r_i < r_{i+1}$
            such that $r \in [r_i+1,r_{i+1}]$ and
            $c(r,p) = c_i(r-r_i,p)$.
            We have found $c_i, p$ and $(r - r_i)$ showing that
            $q$ is in the minimal oblivious strategy for \DEL\@.
            Since $f$ is the minimal oblivious strategy for \DEL,
            we get $q \in f$.
            \qedhere
    \end{itemize}
\end{proof}

\noindent
Combination is different from the other operations, as combining collections
is done round by round. Since oblivious strategies do not depend on the
round, the combination of two oblivious strategies will accept the combination
of any two states accepted, that is, it will accept any intersection of
the delivered set of received messages from the current round in the first state
and the delivered set of received messages from the current round in the second state.
Yet when taking the combination of two predicates, maybe the collections are
such that these two delivered sets used in the intersection above never happen
at the same round, and thus never appear in the combination of collections.

To ensure that every intersection of pairs of delivered sets, one from
a collection from each predicate, happens in the combination of predicates,
we add an assumption: the symmetry of the predicate
over processes and over rounds. This means that for
any delivered set $D$ of the predicate, for any round and any process, there
is a collection of the predicate where $D$ is the delivered set for some round
and some process. $\DEL^{crash}_F$ is an example of round symmetric delivered predicate: all processes are equivalent, and crashes can happen at any round.

\begin{defi}[Round Symmetric \DEL]%
  \label{defRoundSym}
    Let $\DEL$ be a delivered predicate. $\DEL$ is \textbf{round symmetric} iff
    $\forall c \in \DEL, \forall r > 0, \forall k \in
    \Pi, \forall r' > 0, \forall k \in \Pi, \exists c' \in \DEL:
    c(r,k) = c'(r',j)$.
\end{defi}

\begin{thm}[Minimal Oblivious Strategy for Combination]%
    \label{combiObliv}
    Let $\DEL_1, \DEL_2$ be two round symmetric delivered predicates,
    $f_1$ and $f_2$ the minimal
    oblivious strategies for, respectively, $\DEL_1$ and $\DEL_2$.
    Then $f_1 \combi f_2$ is the minimal oblivious strategy for
    $\DEL_1 \combi \DEL_2$.
\end{thm}

\begin{proof}[Proof idea]
    The oblivious states of $\DEL_1 \combi \DEL_2$ are the combination
    of an oblivious state of $\DEL_1$ and of one of $\DEL_2$ at the
    same round, for the same process. Thanks to round symmetry,
    this translates into the combination of any oblivious state of
    $\DEL_1$ with any oblivious state of $\DEL_2$.
    Since $f_1$ and $f_2$ are the minimal oblivious strategy, they both
    accept exactly the oblivious states of $\DEL_1$ and $\DEL_2$
    respectively. Thus $f_1 \combi f_2$ accepts all the combinations
    of oblivious states of $\DEL_1$ and $\DEL_2$, and
    is the minimal oblivious strategy of $\DEL_1 \combi \DEL_2$.
\end{proof}

\begin{proof}
    We show that $f_1 \combi f_2 =
    \{q \mid \exists c \in \DEL_1 \combi \DEL_2, \exists p \in \Pi,
    \exists r > 0: \obliv(q) = c(r,p) \}$, which allows us to apply
    Lemma~\ref{defMinObliv} to show that $f_1 \combi f_2$ is the
    minimal oblivious strategy of $\DEL_1 \combi \DEL_2$.

    \begin{itemize}
        \item
            Let $q \in f_1 \combi f_2$. Then $\exists q_1 \in f_1,
            \exists q_2 \in f_2$ such that $q = q_1 \combi q_2$.
            This also means that $q_1.round = q_2.round = q.round$.
            By minimality of $f_1$ and $f_2$,
            $\exists c_1 \in \DEL_1, \exists  p_1 \in \Pi, \exists r_1 > 0:
            c_1(r_1,p_1) = \obliv(q_1)$ and
            $\exists c_2 \in \DEL_2,\exists p_2 \in \Pi,\exists r_2 > 0:
            c_2(r_2,p_2) = \obliv(q_2)$.
            Moreover, by Definition~\ref{defRoundSym} of round symmetry,
            the hypothesis on $\DEL_2$ ensures that
            $\exists c'_2 \in \DEL_2:c'_2(r_1,p_1)
            = c_2(r_2,p_2)$.
            We take $c = c_1 \combi c'_2$.
            $\obliv(q) = \obliv(q_1) \cap
            \obliv(q_2) = c_1(r_1,p_1) \cap c_2(r_2,p_2)
            = c_1(r_1,p_1) \cap c'_2(r_1,p_1)
            = c(r_1,p_1)$.
            We have $c$, $p_1$ and $r_1$ showing that
            $q$ is in the minimal oblivious strategy for $\DEL_1 \combi \DEL_2$.
        \item Let $q$ be such that $\exists c \in \DEL_1 \combi \DEL_2,
            \exists p \in \Pi, \exists r > 0: c(r,p)= \obliv(q)$.
            By Definition~\ref{defOpsPred} of operations on predicates,
            and specifically of combination, $\exists c_1 \in \DEL_1,
            \exists c_2 \in \DEL_2: c = c_1 \combi c_2$.
            We take $q_1$ such that $q_1.round = r,
            \obliv(q_1)= c_1(r,p)$ and
            $\forall r' \neq r: q_1(r') = q(r')$; we also
            take $q_2$ such that $q_2.round = r,
            \obliv(q_2) = c_2(r,p)$ and
            $\forall r' \neq r: q_2(r') = q(r')$.
            Then $q = q_1 \combi q_2$. By Definition~\ref{defMinObliv} of
            minimal oblivious strategies, $f_1$ and $f_2$ being respectively
            the minimal oblivious strategies of $\DEL_1$ and $\DEL_2$ imply that
            $q_1 \in f_1$ and $q_2 \in f_2$.
            We conclude that $q \in f_1 \combi f_2$.
            \qedhere
    \end{itemize}
\end{proof}

\noindent
This subsection shows that as long as predicates are built from simple
building blocks with known minimal oblivious strategies, the minimal oblivious strategy
of the result can be explicitly constructed.

\subsection{Computing Heard-Of Predicates of Oblivious Strategies}%
\label{sec:oblivious:heardof}

Once the minimal oblivious strategy has been computed, the next step
is to extract the heard-of predicate for this strategy:
the smallest predicate (all its collections are contained in the other predicates)
generated by an oblivious strategy for this delivered predicate.
This ends up being simple: it is the product of all delivered sets.

\begin{defi}[Heard-Of Product]%
  \label{HOProd}
  Let $S \subseteq \mathcal{P}(\Pi)$.
  The \textbf{heard-of product generated by S},
  $\HOProd(S) \triangleq \{h, \textit{a heard-of collection} \mid \forall p \in \Pi,
  \forall r > 0: h(r,p) \in S \}$.
\end{defi}

Here is the intuition: defining a heard-of collection requires,
for each round and each process, the corresponding heard-of set.
A heard-of product is then the set of all collections
that have heard-of sets from the set given as argument.
So the total heard-of predicate (containing only the total collection) is the
heard-of product of the set $\Pi$. And $\HO_F$ is the heard-of
product of all subsets of $\Pi$ of size $\geq n-F$.

The following lemma links the $\Nexts_f$
of some valid oblivious strategy and the heard-of predicate for
this strategy: the predicate is the heard-of product of the $\Nexts_f$.

\begin{lem}[Heard-Of Predicate of an Oblivious Strategy]%
  \label{carefrHO}
  Let $\DEL$ be a delivered predicate containing $c_{total}$ and
  let $f$ be a valid oblivious strategy for \DEL\@. Then
  $\HO_f(\DEL) = \HOProd(\Nexts_f)$.
\end{lem}

\begin{proof}\hfill
  \begin{itemize}
    \item $(\subseteq)$ To prove this first direction, we show
      that the heard-of sets of any collection in
      $\HO_f(\DEL)$ are in $\Nexts_f$. This then entails
      that $\HO_f(\DEL) \subseteq \HOProd(\Nexts_f)$.
      By Definition~\ref{defHOExec} of the heard-of
      collection of an execution, every heard-of set contains the
      set of messages from the current round that was already received
      at the $next_p$ transition where the process $p$ changed round. By
      Definition~\ref{defExecStrat} of the executions of a strategy,
      such a $next_p$ transition can only happen if the local state
      of the process $p$ is in $f$. By Definition~\ref{defObliv} of
      oblivious strategies, $f$ contains exactly the states such that
      the messages received from the current round form a set
      in $\Nexts_f$.
      Therefore, the heard-of set of any collection generated by
      $f$ on a collection of $\DEL$ are necessarily
      in $\Nexts_f$.
    \item $(\supseteq)$ Let \textit{ho} be a heard-of collection
      such that $\forall r > 0, \forall j \in \Pi: \textit{ho}(r,j) \in
      \Nexts_f$. Let $t$ be the canonical execution of $ho$. It is
      an execution by Lemma~\ref{trivial}. By Definition~\ref{defCanExec}, at each round,
      processes receive a set of messages in $\Nexts_f$. By Definition~\ref{defObliv}
      of oblivious strategies, this entails that the local states are in $f$ when the processes change rounds. And so, by Definition~\ref{defExecStrat}, it is an execution of $f$.
      Hence $t$ is an execution of $f$ on \DEL, because $DEL$ contains the total predicate. Since $t = can(ho)$,
      Lemma~\ref{canHO} implies that $h_t = ho$.
      We conclude that $\textit{ho} \in \HO_f(\DEL)$.
      \qedhere
    \end{itemize}
\end{proof}

\noindent
Thanks to this characterization, the heard-of predicate
generated by the minimal strategies for the operations is computed in terms
of the heard-of predicate generated by the original minimal strategies.

\begin{thm}[Heard-Of Predicate of Minimal Oblivious Strategies]%
  \label{oblivOpsHO}
  Let $\DEL, \DEL_1, \DEL_2$ be delivered predicates containing $c_{total}$.
  Let $f, f_1, f_2$ be their respective minimal oblivious strategies.
  Then:
  \begin{itemize}
    \item $\HO_{f_1 \cup f_2}(\DEL_1 \cup \DEL_2)
      = \HO_{f_1 \cup f_2}(\DEL_1 \leadsto \DEL_2)
      = \HOProd(\Nexts_{f_1} \cup \Nexts_{f_2})$.
    \item If $\DEL_1$ or $\DEL_2$ are round symmetric, then:\\
      $\HO_{f_1 \combi f_2}(\DEL_1 \combi \DEL_2)=
      \HOProd(\{ n_1 \cap n_2 \mid
      n_1 \in \Nexts_{f_1} \land n_2 \in \Nexts_{f_2}\})$.
    \item $\HO_f(\DEL^{\omega}) = \HO_f(\DEL)$.
  \end{itemize}
\end{thm}

\begin{proof}
    Obviously, we want to apply Lemma~\ref{carefrHO}. Then
    we first need to show that the delivered predicates contain $c_{total}$.
    \begin{itemize}
        \item By hypothesis, $\DEL_1$ and $\DEL_2$ contain $c_{total}$.
            Then $\DEL_1 \cup \DEL_2$ trivially contains it too by
            Definition~\ref{defOpsPred} of operations on predicates.
        \item By hypothesis, $\DEL_1$ and $\DEL_2$ contain $c_{total}$.
            Then $\DEL_1 \combi \DEL_2$ contains
            $c_{total} \combi c_{total} = c_{total}$ by
            Definition~\ref{defOpsPred} of operations on predicates.
        \item By hypothesis, $\DEL_1$ and $\DEL_2$ contain $c_{total}$.
            Then $\DEL_1 \leadsto \DEL_2 \supseteq \DEL_2$
            contains it too by Definition~\ref{defOpsPred}
            of operations on predicates.
        \item By hypothesis, $\DEL$ contains $c_{total}$.
            We can recreate $c_{total}$ by taking all $c_i = c_{total}$
            and whichever $r_i$. Thus $\DEL^{\omega}$ contains
            $c_{total}$ by Definition~\ref{defOpsPred} of operations on predicates.
    \end{itemize}

    \noindent
    Next, the strategies $f_1 \cup f_2$, $f_1 \combi f_2$ and $f$ are the
    respective minimal oblivious strategies by Theorem~\ref{unionObliv},
    Theorem~\ref{combiObliv} and Theorem~\ref{repetObliv}. They are also
    valid by Theorem~\ref{oblivValid}.

    Lastly, we need to show that the $\Nexts_f$ for the strategies
    corresponds to the generating sets in the theorem.
    \begin{itemize}
        \item We show $\Nexts_{f_1 \cup f_2} =
            \Nexts_{f_1} \cup \Nexts_{f_2}$,
            and thus that
            \[\HOProd(\Nexts_{f_1 \cup f_2}) =
            \HOProd(\Nexts_{f_1} \cup \Nexts_{f_2}).\]
            \begin{itemize}
                \item $(\subseteq)$ Let $n \in \Nexts_{f_1 \cup f_2}$. Then $\exists
                    q \in f_1 \cup f_2: \obliv(q) = n$. By Definition~\ref{defOpsPred}
                    of operations on predicates, and specifically union,
                    $q \in f_1$ or $q \in f_2$. We fix $q \in f_1$
                    (the case $q \in f_2$ is symmetric).
                    Then $n \in \Nexts_{f_1}$ by Definition~\ref{defObliv} of
                    oblivious strategies.
                    We conclude that
                    $n \in \Nexts_{f_1} \cup \Nexts_{f_2}$.
                \item $(\supseteq)$ Let $n \in \Nexts_{f_1} \cup \Nexts_{f_2}$.
                    We fix $n \in \Nexts_{f_1}$
                    (as always, the other case is symmetric).
                    Then $\exists q \in f_1: \obliv(q) = n$. As
                    $q \in f_1$ implies $q \in f_1 \cup f_2$,
                    we conclude that $n \in \Nexts_{f_1 \cup f_2}$
                    by Definition~\ref{defObliv} of oblivious strategies.
            \end{itemize}
        \item We show $\Nexts_{f_1 \combi f_2} = \{ n_1 \cap n_2 \mid
            n_1 \in \Nexts_{f_1} \land n_2 \in \Nexts_{f_2}\}$,
            and thus that
            $\HOProd(\Nexts_{f_1 \combi f_2}) =
            \HOProd(\{ n_1 \cap n_2 \mid
            n_1 \in \Nexts_{f_1} \land n_2 \in \Nexts_{f_2}\})$.
            \begin{itemize}
                \item $(\subseteq)$ Let $n \in \Nexts_{f_1 \combi f_2}$. Then
                    $\exists q \in f_1 \combi f_2: \obliv(q) = n$.
                    By Definition~\ref{defOpsPred} of operations on predicates,
                    and specifically of combination, $\exists q_1 \in f_1,
                    \exists q_2 \in f_2: q_1.round = q_2.round= q.round
                    \land q = q_1 \combi q_2$. This means
                    $n = \obliv(q) = \obliv(q_1) \cap \obliv(q_2)$.
                    We conclude that $n \in \{ n_1 \cap n_2 \mid
                    n_1 \in \Nexts_{f_1} \land n_2 \in \Nexts_{f_2}\}$
                    by Definition~\ref{defObliv} of oblivious strategies.
                \item $(\supseteq)$ Let $n \in \{ n_1 \cap n_2 \mid
                    n_1 \in \Nexts_{f_1} \land n_2 \in \Nexts_{f_2}\}$.
                    Then $\exists n_1 \in \Nexts_{f_1},
                    \exists n_2 \in \Nexts_{f_2}: n = n_1 \cap n_2$.
                    By Definition~\ref{defObliv} of oblivious strategies,
                    and because $f_1$ and $f_2$ are oblivious strategies, we
                    can find $q_1 \in f_1$ such that $\obliv(q_1) = n_1$,
                    $q_2 \in f_2$ such that $\obliv(q_2) = n_2$, and
                    $q_1.round = q_2.round$.
                    Then $q = q_1 \combi q_2$ is a state
                    of $f_1 \combi f_2$. We have $\obliv(q) =
                    n_1 \cap n_2 = n$.
                    We conclude that $n \in \Nexts_{f_1 \combi f_2}$
                    by Definition~\ref{defObliv} of oblivious strategies.
            \end{itemize}
        \item Trivially, $\Nexts_f = \Nexts_f$.
        \qedhere
    \end{itemize}
\end{proof}

\subsection{Domination by Oblivious Strategies}%
\label{sec:oblivious:domination}

Finally, the value of oblivious strategies depends on which delivered
predicates have a dominating oblivious strategy. $\DEL^{crash}_F$ does, with the strategy $f_{n-F}$.
We consider delivered predicates
that satisfy the so-called common round property.
This condition captures the fact that given any delivered set $D$,
one can build, for any $r > 0$, a delivered collection where processes
receive all the messages up to round $r$, and then they share $D$ as their
delivered set in round $r$. As a limit case, the predicate also contains
the total collection. This common round property is preserved by the composition operators, which allows to derive complex dominating oblivious strategies from simple ones.

\begin{defi}[Common Round Property]%
  \label{defCommonRound}
  Let $\DEL$ be a delivered predicate.
  $\DEL$ has the \textbf{common round property} $\triangleq$
  \begin{itemize}
    \item \textbf{(Total collection)}
      $\DEL$ contains the total collection $c_{total}$.
    \item \textbf{(Common round)} $\forall c \in \DEL,
      \forall r > 0, \forall j \in \Pi,
      \forall r' > 0,
      \exists c' \in \DEL, \forall p \in \Pi:
      (\forall r'' < r': c'(r'',p) = \Pi \land c'(r',p) = c(r,j))$.
  \end{itemize}
\end{defi}

\noindent
What the common round property captures is what makes $\DEL^{crash}_F$
be dominated by an oblivious strategy: if one process $j$ might block
at round $r$ even after receiving all messages from round $r$
in some collection $c$ of $\DEL$, and all messages from rounds $< r$,
then there is a collection and an execution where all processes
block in the same way. The collection ensures that the delivered collection
gives each process the same delivered sets ($\Pi$) for rounds $< r$,
and $c(r,j)$ at round $r$. The execution is the standard execution
of this collection, that puts every process at round $r$ in the same
blocking state as $j$, and so a deadlock occurs.
The conclusion is that any valid strategy should allow
to change round when all messages from previous rounds are received,
and the messages received for the current round form a delivered set from
a collection of \DEL\@.
Applying this reasoning to the canonical executions of heard-of
collections from $\HO_{f_{\min}}(\DEL)$ yields that the canonical
executions are executions of any valid strategy for $\DEL$ (not
only oblivious ones), and thus that for any valid strategy $f$ for
$\DEL$, $\HO_{f_{\min}}(\DEL) \subseteq \HO_f(\DEL)$. That is to say,
$\DEL$ is dominated by an oblivious strategy.

\begin{thm}[Sufficient Condition of Oblivious Domination]%
  \label{suffCfreeDom}
  Let $\DEL$ be a delivered predicate satisfying the common round property.
  Then, there is an oblivious strategy that dominates \DEL\@.
\end{thm}

\begin{proof}
  Let $f_{\min}$ be the minimal oblivious strategy for \DEL\@.
  It dominates the oblivious strategies for $\DEL$ by Lemma~\ref{domMinObliv}.
  We now prove that $f_{\min}$ dominates \DEL\@. This amount to showing that
  for $f'$, a valid strategy for \DEL, we have
  $f' \prec_{\DEL} f_{\min}$, that is
  $\HO_{f_{\min}}(\DEL) \subseteq \HO_{f'}(\DEL)$.
  Let $\textit{ho} \in \HO_{f_{\min}}(\DEL)$
  and $t$ be the canonical execution of \textit{ho}. We show that $t$ is an
  execution of $f'$, which entails by Lemma~\ref{canHO} that
  $\textit{ho} \in \HO_{f'}(\DEL)$.

  By Definition~\ref{defCommonRound} of the common round property, $\DEL$
  contains $c_{total}$. By Lemma~\ref{trivial},
  $t$ is an execution of $c_{total}$, and thus an execution of \DEL\@.
  We now prove by contradiction
  it is also an execution of $f'$ on \DEL\@. \textbf{Assume} it is not.
  By Definition~\ref{defExecStrat} of the executions of a strategy,
  the problem comes either from breaking fairness or from some $next_j$ for some
  process $j$ at a point where the local state of $j$ is not in $f'$.
  Since for every $j \in \Pi: next_j$ happens
  an infinite number of times in $t$ by Definition~\ref{defCanExec} of
  a canonical execution, the only possibility left is the second one: some $next_j$
  in $t$ is done while the local state of $j$ is not in $f'$.
  There thus exists a smallest $r$ such that some process $j$ is not allowed
  by $f'$ to change round when $next_j$ is played at round $r$
  in $t$.

  Lemma~\ref{carefrHO} yields that
  $\HO_{f_{\min}}(\DEL) = \HOProd(\Nexts_{f_{\min}})$.
  By Definition~\ref{defMinObliv} of minimal oblivious strategies,
  $\Nexts_{f_{\min}} =
  \{ c(r',p) \mid c \in \DEL \land r' > 0 \land p \in \Pi\}$.
  Thus $\exists c \in \DEL, \exists r' > 0, \exists p \in \Pi:
  \textit{ho}(r,j) = c(r',p)$.
  Since $\DEL$ satisfies the common round property (Definition~\ref{defCommonRound}), it allows us to build
  $c_{block} \in \DEL$ such that
  $\forall r'' < r, \forall k \in \Pi:
  c_{block}(r'',k) = \Pi$ and
  $\forall k \in \Pi:
  c_{block}(r,k) = c(r',j) = ho(r,j)$.
  Finally, we build $t_{block} = st(f',c_{block})$ the standard execution
  of $f'$ on $c_{block}$.
  By Lemma~\ref{stCorrect}, we know $t_{block}$ is an execution of $f'$ on $c_{block}$.
  We then show that it is invalid by examining the two possibilities.
  \begin{itemize}
    \item During one of the first $r-1$ iterations of
      $t_{\textit{block}}$, there is some process that cannot change
      round. Let $r'$ be the smallest iteration where it happens,
      and $k$ be a process unable to change round at the $r'$-th
      iteration.
      By minimality of $r'$, all the processes arrive at round~$r'$,
      and by definition of $c_{\textit{block}}$ they all
      receive the same messages as $k$ before $changes_{r'}$.
      That means every process has the same local state as $k$.
      Thus, all the processes are blocked at round $r'$,
      there are no more round changes or deliveries,
      and $t_{block}$ is invalid by Definition~\ref{valid} of
      validity.
    \item For the first $r-1$ iterations, every process changes round.
      Thus everyone arrives at round~$r$. By Definition~\ref{defStdExec} of
      the standard execution, all messages from round $r$ are delivered
      before the $change_{r}$ section. The definition of
      $c_{\textit{block}}$ also ensures that every process
      received the same messages, that is all the messages from
      round $< r$ and all the messages from \textit{ho}$(r,j)$.
      These are the messages received by $j$ in $t$ at
      round $r$.
      By hypothesis, $j$ is blocked in this state in $t$.
      We thus deduce that all the processes are blocked at round $r$
      in $t_{block}$, and that $t_{block}$ is an invalid execution by
      Definition~\ref{valid} of validity.
  \end{itemize}
  Either way, we deduce that $f'$ is invalid, which is a
  \textbf{contradiction}.
  We conclude that $t$ is an execution of $f'$ on \DEL\@. Lemma~\ref{canHO}
  therefore implies that $\textit{ho} \in \HO_{f'}(\DEL)$.
  This entails that
  $\HO_{f_{\min}}(\DEL) \subseteq \HO_{f'}(\DEL)$, and thus
  that $f' \prec_{\DEL} f_{\min}$. We conclude that $f_{\min}$ dominates
  $\DEL$ by Definition~\ref{defDom}.
\end{proof}

This condition is maintained by the operations. Hence
any predicate built from ones satisfying this condition will still be
dominated by an oblivious strategy.

\begin{thm}[Domination by Oblivious for Operations]%
  \label{invarCommon}
  Let $\DEL, \DEL_1, \DEL_2$ be delivered predicates satisfying the
  common round property.
  Then $\DEL_1 \cup \DEL_2$, $\DEL_1 \combi \DEL_2$, $\DEL_1 \leadsto
  \DEL_2$, $\DEL^{\omega}$ also satisfy the common round property.
\end{thm}

\begin{proof}
  Thanks to Theorem~\ref{suffCfreeDom}, we have to show that
  the condition is maintained by the operations; the domination
  by an oblivious strategy directly follows from Theorem~\ref{suffCfreeDom}.
  The fact that $c_{total}$ is still in the results of
  the operations was already shown in the proof of Theorem~\ref{oblivOpsHO}.
  Hence we show the invariance of the common round part.
  \begin{itemize}
    \item Let $c \in \DEL_1 \cup \DEL_2$. Thus $c \in \DEL_1$
      or $c \in \DEL_2$. We fix $c \in \DEL_1$ (the other
      case is symmetric). Then for $p \in \Pi,r > 0$ and $r' > 0$,
      we get a $c' \in \DEL_1$
      satisfying the condition of Definition~\ref{defCommonRound}
      by the hypothesis that $\DEL_1$ satisfies the common round
      property.
      Since $\DEL_1 \subseteq \DEL_1 \cup \DEL_2$,
      we get $c' \in \DEL_1 \cup \DEL_2$.
      We conclude that the condition still holds
      for $\DEL_1 \cup \DEL_2$.
    \item Let $c \in \DEL_1 \combi \DEL_2$. Then
      $\exists c_1 \in \DEL_1, \exists c_2 \in \DEL_2:
      c = c_1 \combi c_2$. For $p \in \Pi,r > 0$ and $r' > 0$,
      our hypothesis on $\DEL_1$ and $\DEL_2$ ensures
      that there are $c'_1 \in \DEL_1$ satisfying the
      condition of Definition~\ref{defCommonRound} for $c_1$
      and $c_2' \in \DEL_2$ satisfying
      the condition of Definition~\ref{defCommonRound} for $c_2$.
      We argue that $c' = c_1' \combi c'_2$ satisfies
      the condition of Definition~\ref{defCommonRound} for $c$. Indeed, $\forall r'' < r',
      \forall q \in \Pi: c(r'',q) = c_1'(r'',q) \combi c_2'(r'',q)
      = \Pi$ and $\forall q \in \Pi: c(r',q) = c_1'(r',q) \combi
      c_2'(r',q) = c_1(r,p) \combi c_2(r,p) = c(r,p)$.
      We conclude that the condition of Definition~\ref{defCommonRound}
      still holds for $\DEL_1 \combi \DEL_2$.
    \item Let $c \in \DEL_1 \leadsto \DEL_2$. Since if $c \in
      \DEL_2$ the condition of Definition~\ref{defCommonRound}
      trivially holds by hypothesis,
      we study the case where succession actually happens.
      Hence $\exists c_1 \in \DEL_1, \exists c_2 \in \DEL_2, \exists
      r_{change} > 0: c = c_1[1,r_{change}].c_2$.
      For $p \in \Pi, r > 0$ and $r'>0$, we consider two cases.
      \begin{itemize}
        \item if $r \leq r_{change}$, then
          our hypothesis on $\DEL_1$ ensures
          that there is $c'_1 \in \DEL_1$ satisfying the
          condition of Definition~\ref{defCommonRound} for $c_1$.
          We argue that
          $c' = c_1'[1,r'].c_2 \in \DEL_1 \leadsto \DEL_2$
          satisfies the condition of Definition~\ref{defCommonRound} for $c$.
          Indeed, $\forall r'' < r', \forall q \in \Pi:
          c'(r'',q) = c_1'(r'',q) = \Pi$, and $\forall q \in \Pi:
          c'(r',q) = c_1(r,p) = c(r,p)$.
        \item if $r > r_{change}$, then
          our hypothesis on $\DEL_2$ ensures
          that there is $c'_2 \in \DEL_2$ satisfying
          the condition of Definition~\ref{defCommonRound}
          for $c_2$ at $p$ and $r - r_{change}$.
          That is, $c'_2[1,r'-1] = c_{total}[1,r'-1] \land
          \forall q \in \Pi: c'_2(r',q)=c_2(r-r_{change},p)$
          We argue that $c' = c'_2 \in \DEL_1 \leadsto \DEL_2$
          satisfies the condition of Definition~\ref{defCommonRound} for $c$.
          Indeed, $\forall r'' < r', \forall q \in \Pi:
          c'_2(r'',q) = \Pi$, and $\forall q \in \Pi:
          c'_2(r',q) = c_2(r-r_{change},p) = c(r,p)$.
      \end{itemize}
      We conclude that the condition of Definition~\ref{defCommonRound} still holds for
      $\DEL_1 \leadsto \DEL_2$.
    \item Let $c \in \DEL^{\omega}$. Let $(c_i)$ and
      $(r_i)$ be the collections and indices defining
      $c$. We take $p \in \Pi,r > 0$ and $r' > 0$.
      Let $i > 0$ be the integer such that
      $r \in [r_i+1, r_{i+1}]$. By hypothesis on $\DEL$,
      there is $c'_i \in \DEL$ satisfying the condition of Definition~\ref{defCommonRound}
      for $c_i$ at $p$ and $r - r_i$.
      That is, $c'_i[1,r'-1] = c_{total}[1,r'-1] \land
      \forall q \in \Pi: c'_i(r',q)=c_i(r-r_i,p)$.
      We argue that $c'_i \in \DEL$
      satisfies the condition of Definition~\ref{defCommonRound} for $c$. Indeed,
      $\forall r'' \leq r', \forall q \in \Pi$, we have:
      $c'_i(r'',q) = \Pi$ and $\forall q \in \Pi:
      c'_i(r',q) = c_i(r-r_i,p) = c(r,p)$.
      We conclude that the condition of Definition~\ref{defCommonRound} still holds for
      $\DEL^{\omega}$.
      \qedhere
  \end{itemize}
\end{proof}

  \noindent
  Therefore, as long as the initial building blocks satisfy the common round
  property, so do the results of the operations --- and the latter
  is dominated by its minimal oblivious strategy, a strategy that can be computed
  easily from the results of this section.

\section{Conservative Strategies: Looking at Present and Past Rounds}%
\label{sec:chap2conserv}

The class of considered strategies expands the class of oblivious strategies
by considering past rounds and the round number in addition to the present
round. This is a generalization of oblivious strategies that trades
simplicity for expressivity, while retaining a nice structure.

The structure of this section is similar to the previous one: it defines conservative strategies and give a necessary and sufficient condition for a conservative strategy to be valid (Section~\ref{sec:conservative:def}), presents results on the composition of conservative strategies (Section~\ref{sec:conservative:composition}) that enables to compute upper bounds on the heard-of predicates of conservative strategies (Section~\ref{sec:conservative:heardof}), and ends with a sufficient condition for conservative domination, condition preserved by composition (Section~\ref{sec:conservative:domination}).

\subsection{Definition and Expressiveness}%
\label{sec:conservative:def}

\begin{defi}[Conservative Strategy]%
  \label{defCons}
  Let $\cons$ be the function such that \[\forall q \in Q,\ \cons(q) \triangleq
  \langle q.round, \{ \langle r, k \rangle \in q.mes \mid r \le q.round\}\rangle.\]
  Let $\approx_{\cons}$ the equivalence relation defined by
  $q_1 \approx_{\cons} q_2$ if $\cons(q_1) = \cons(q_2)$.
  The family of \textbf{conservative strategies} is $\textit{family}(\approx_{\cons})$.
  We write $\Nexts^C_f \triangleq \{\cons(q) \mid q \in f\}$ for the set
  of conservative states in $f$. This uniquely defines $f$.
\end{defi}

In analogy with the case of oblivious strategies, there is an intuitive
necessary and sufficient condition for such a strategy to be valid for a given
delivered predicate.

\begin{lem}[Necessary and Sufficient Condition for Validity of a Conservative Strategy]%
  \label{consValid}
  Let $\DEL$ be a delivered predicate and $f$ be a conservative strategy.
  Then $f$ is valid for $\DEL \iff
  f \supseteq \{q \in Q \mid \exists c \in \DEL, \exists p \in \Pi,
  \forall r \leq q.round: q(r) = c(r,p) \}$.
\end{lem}

\begin{proof}\hfill
  \begin{itemize}
    \item $(\Rightarrow)$ Let $f$ be valid for \DEL\@.
      We show by contradiction that it satisfies the right-hand side of the
      above equivalence.
      \textbf{Assume} there is $q_{block}$ a local state such that
      $\exists c \in \DEL, \exists r > 0, \exists j \in \Pi:
      \cons(q_{block}) = \langle r, \{ \langle r',k \rangle \mid r' \leq r
      \land k \in c(r',j) \} \rangle$ and $q \notin f$.
      By Definition~\ref{defCons}, this means that for every $q$ such that
      $\cons(q) = \cons(q_{block}) =
      \langle r, \{ \langle r',k \rangle \mid r' \leq r \land k \in c(r',j) \} \rangle$,
      $q \notin f$.
      Let $t = st(f,c)$ be the standard execution of $f$ on $c$.
      This is an execution of $f$ on $t$ by Lemma~\ref{stCorrect}.
      The sought contradiction is reached by proving that $t$ is invalid for $\DEL$,
      and thus $f$ is invalid for $\DEL$ too.
      To do so, we split according to two cases: the first is the case where
      there is a blocking process before round $r$, and the other uses the
      hypothesis on the prefix of $c$ for $j$ up to round $r$.
      \begin{itemize}
          \item During one of the first $r-1$ iterations of
              $t$, there is some process which cannot change
              round. Let $r'$ be the smallest iteration of the canonical
              execution where it happens,
              and $k$ be a process unable to change round at the $r'$-th
              iteration.
              By minimality of $r'$, all the processes arrive at round $r'$ in
              $t$; by Definition~\ref{defStdExec} of the standard execution,
              all messages for $k$
              from all rounds up to $r'$ are delivered before the $change$ part
              of the iteration.
              Let $q$ the local state of $k$ at the start of
              $change_{r'}$, and let $q'$ be any local state of $k$ afterward.
              The above tells us that as long as $q'.round = q.round$,
              we have $\cons(q) = \cons( q')$ and thus $q' \notin f$.
              Therefore, $k$ can never change round while at round $r'$.
              We conclude that $t$ is invalid for $\DEL$ by Definition~\ref{valid}.
          \item For the first $r-1$ iterations, all the processes change round.
              Thus everyone arrives at round $r$ in the $r-1$-th iteration.
              By Definition~\ref{defStdExec} of the standard execution,
              all messages from rounds up to $r$ are delivered
              before the $change_r$ part of the $r$-th iteration.
              Thus $j$ is in a local state $q$
              at the $change_{r}$ part of the $r$-th iteration
              such that $\cons(q) = \langle r, \{ \langle r',k \rangle \mid
              r' \leq r \land k \in c(r',j) \} \rangle =
              \cons(q_{block})$. By hypothesis, this means $q \notin f$
              thus that $j$ cannot change round.
              Let $q'$ be any local state of $j$ afterward.
              The above tells us that as long as $q'.round = q.round$,
              we have $\cons(q) = \cons(q')$
              and thus $q' \notin f$. Therefore, $j$ can never
              change round while at round $r$.
              Here too, $t$ is invalid for $\DEL$ by Definition~\ref{valid}.
      \end{itemize}
      Either way, we reach a \textbf{contradiction} with the validity of $f$ for \DEL\@.
    \item $(\Leftarrow)$ Let $\DEL$ and $f$ be such that
      $\forall c \in \DEL,
      \langle r, \{ \langle r',k \rangle \mid r' \leq r
      \land k \in c(r',j) \} \rangle \in \Nexts^C_f$.
      We show by contradiction that $f$ is valid for \DEL\@.
      \textbf{Assume} the contrary: there is some
      $t \in \execs_f(\DEL)$ that is invalid for \DEL\@. Thus there
      are some processes blocked at a round forever in $t$. Let $r$ be the
      smallest such round, and $j$ be a process blocked at round $r$ in $t$.
      By minimality of $r$, all the processes arrive at round $r$.
      By Definition~\ref{defExecColl} of an execution of \DEL, there is a
      $c \in \DEL$ such that $t$ is an execution of $c$.
      This means by Definition~\ref{defExecColl} of an execution of a collection
      that all messages from all the delivered sets of $j$ up to round
      $r$ are eventually delivered.
      From this point on, every local state $q$ of $j$ satisfies
      $\cons(q) = \langle r, \{ \langle r',k \rangle \mid
      r' \leq r \land k \in c(r',j) \} \rangle$;
      thus we have $q \in f$ by hypothesis on $f$.
      Then the fairness condition of executions of $f$ from
      Definition~\ref{defExecStrat} imposes that $j$ does change round
      at some point.
      We conclude that $j$ is not blocked at round $r$ in $t$,
      which \textbf{contradicts} the hypothesis that $j$ is blocked forever
      at round $r$ in $t$.
      \qedhere
  \end{itemize}
\end{proof}

\noindent
The strategy satisfying exactly this condition is the minimal conservative strategy of
$\DEL$, and it is a strategy dominating all the conservative strategies for
this delivered predicate.

\begin{defi}[Minimal Conservative Strategy]%
  \label{minCons}
  Let $\DEL$ be a delivered predicate.
  Then the \textbf{minimal conservative strategy}
  for $\DEL$ is $f_{\min} \triangleq$ the conservative strategy
  such that $f =
  \{q \in Q \mid \exists c \in \DEL, \exists p \in \Pi,
  \forall r \leq q.round: q(r) = c(r,p) \}$.
\end{defi}

Intuitively, when every message from a prefix is delivered, there is
no message left from past and present; a valid
conservative strategy has to accept the state, or it would be blocked forever.

\begin{rem}[Prefix and conservative state of a prefix]
  Intuitively, a prefix of a collection~$c$
  for a process $p$ at round $r$ is the sequence of sets
  of messages received by $p$ at rounds $\leq r$ in $c$. Then
  we can define a state corresponding to this prefix by fixing its
  round at $r$ and adding to it the messages in the prefix.
  This is the conservative state of the prefix.
  The prefixes of a delivered predicate are then
  all the prefixes of all its collections.
\end{rem}

\begin{lem}[Domination of Minimal Conservative Strategy]%
  \label{minDom}
  Let $\DEL$ be a delivered predicate and $f_{\min}$
  be its minimal conservative strategy.
  Then $f_{\min}$ dominates the conservative strategies for \DEL\@.
\end{lem}

\begin{proof}
  First, $f_{\min}$ is valid for $\DEL$ by application of
  Lemma~\ref{consValid}.
  Next, we take another conservative strategy $f$, valid
  for \DEL\@. Lemma~\ref{consValid} gives us that
  $f_{\min} \subseteq f$. Hence, when
  $f_{\min}$ allow a change of round, so does $f$.
  This entails by Definition~\ref{defExecStrat}
  that all the executions of $f_{\min}$ for
  $\DEL$ are also executions of $f$ for $\DEL$,
  and by Definition~\ref{defHOExec}
  that $\HO_{f_{\min}}(\DEL) \subseteq
  \HO_f(\DEL)$.
  We conclude from Definition~\ref{defDom}
  that $f_{\min}$ dominates any valid conservative strategy for \DEL\@.
\end{proof}

\subsection{Composition of Conservative Strategies}%
\label{sec:conservative:composition}

Like oblivious strategies, applying operations to minimal conservative strategies
gives the minimal conservative strategies of the predicates after the analogous operations.

\begin{thm}[Minimal Conservative Strategy for Union]%
    \label{unionCons}
    Let $\DEL_1, \DEL_2$ be two delivered predicates,
    $f_1$ and $f_2$ the minimal
    conservative strategies for, respectively, $\DEL_1$ and $\DEL_2$.
    Then $f_1 \cup f_2$ is the minimal conservative strategy for
    $\DEL_1 \cup \DEL_2$.
\end{thm}

\begin{proof}
    We only have to show that $f_1 \cup f_2$ is equal to
    Definition~\ref{minCons}.
    \begin{itemize}
        \item $(\supseteq)$ Let $q$ be a state such that
            $\exists c \in \DEL_1 \cup \DEL_2, \exists p \in \Pi$ such
            that $\forall r \leq q.round: q(r) = c(r,p)$.
            If $c \in \DEL_1$, then $q \in f_1$,
            by Definition~\ref{minCons} of the minimal conservative strategy
            because $f_1$ is the minimal conservative strategy
            for $\DEL_1$, and by application of Lemma~\ref{consValid}.
            Thus $q \in f_1 \cup f_2$.
            If $c \in \DEL_2$, the same reasoning applies with
            $f_2$ in place of $f_1$.
            We conclude that $q \in f_1 \cup f_2$.
        \item $(\subseteq)$ Let $q \in f_1 \cup f_2$.
            This means that $q \in f_1 \lor q \in f_2$. The case where
            it is in both can be reduced to any of the two.
            If $q \in f_1$, then by Definition~\ref{minCons} of the minimal
            conservative strategy and by minimality of $f_1$,
            $\exists c_1 \in \DEL_1, \exists p_1 \in \Pi$ such that
            $\forall r \leq q.round: q(r) = c_1(r,p_1)$.
            $\DEL_1 \subseteq
            \DEL_1 \cup \DEL_2$, thus $c_1 \in \DEL_1 \cup \DEL_2$.
            We have found the $c$ and $p$ necessary to show $q$ is
            in the minimal conservative strategy for $\DEL_1 \cup \DEL_2$.
            If $q \in f_2$, the reasoning is similar to
            the previous case, replacing $f_1$ by $f_2$
            and $\DEL_1$ by $\DEL_2$.
            \qedhere
    \end{itemize}
\end{proof}

\noindent
For the other three operations, slightly more structure is needed on
the predicates. More precisely, they have to be independent of
the processes. We require that any prefix of a process $k$ in a collection
of the predicate is also the prefix of any other process $j$
in a possibly different collection of the same \DEL\@. Hence
the behaviors (fault, crashes, loss) are not targeting specific
processes.
This restriction fits the intuition behind many common fault models.

\begin{defi}[Prefix Symmetric \DEL]%
  \label{defPrefixSym}
    Let $\DEL$ be a delivered predicate. $\DEL$ is \textbf{prefix symmetric} if
    $\forall c \in \DEL, \forall k \in \Pi, \forall r > 0,
    \forall j \in \Pi, \exists c' \in \DEL, \forall r' \leq r:
    c'(r',k) = c(r',j)$.
\end{defi}

This differs from the previous round symmetric $\DEL$, in that here we focus on
prefixes, while the other focused on rounds. Notice that none implies the other:
round symmetry says nothing about the rest of the prefix, and prefix symmetry
says nothing about the delivered sets when rounds are different.
Assuming prefix symmetry, the conservation of the minimal conservative strategy
by combination, succession and repetition follows.

\begin{thm}[Minimal Conservative Strategy for Combination]%
    \label{combiCons}
    Let $\DEL_1, \DEL_2$ be two prefix symmetric delivered predicates,
    $f_1$ and $f_2$ the minimal
    conservative strategies for, respectively, $\DEL_1$ and $\DEL_2$.
    Then $f_1 \combi f_2$ is the minimal conservative strategy for
    $\DEL_1 \combi \DEL_2$.
\end{thm}

\begin{proof}[Proof idea]
    Since $f_1$ and $f_2$ are the minimal conservative strategies
    of $\DEL_1$ and
    $\DEL_2$, $\Nexts^C_{f_1}$ is the set of the conservative states of
    prefixes of $\DEL_1$ and $\Nexts^C_{f_2}$ is the set of the conservative
    states of prefixes of $\DEL_2$.
    Also, the states accepted by $f_1 \combi f_2$ are the combination
    of the states accepted by $f_1$ and the states accepted by $f_2$.
    The prefixes of $\DEL_1 \combi \DEL_2$ are the prefixes of
    $\DEL_1$ combined with the prefixes of $\DEL_2$ for
    the same process. Thanks to prefix symmetry,
    we can take a prefix of $\DEL_2$ and any process, and find
    a collection such that the process has that prefix.
    Therefore, the combined prefixes for the same process are the
    same as the combined prefixes of $\DEL_1$ and $\DEL_2$.
    Thus $\Nexts^C_{f_1 \combi f_2}$ is the set of conservative states
    of prefixes of $\DEL_1 \combi \DEL_2$,
    and $f_1 \combi f_2$ is its minimal conservative strategy.
\end{proof}

\begin{proof}
    We only need to show that $f_1 \combi f_2$ is equal to
    Definition~\ref{minCons}.
    \begin{itemize}
        \item $(\supseteq)$ Let $q$ be a state such that
            $\exists c \in \DEL_1 \combi \DEL_2, \exists p \in \Pi$ such
            that $\forall r \leq q.round: q(r) = c(r,p)$.
            By definition of $c$, $\exists c_1 \in \DEL_1,
            \exists c_2 \in \DEL_2: c_1 \combi c_2 = c$.
            We take $q_1$ such that
            \begin{mathpar}
            q_1.round = q.round
            \and \text{and} \and
            \forall r > 0:
            \left(
            \begin{array}{ll}
                q_1(r) = c_1(r,p) &\text{if }r \leq q.round\\
                q_1(r) = q(r) &\text{otherwise}\\
            \end{array}
            \right).
            \end{mathpar}
            We also take $q_2$ such that
            \begin{mathpar}
            q_2.round = q.round
            \and \text{and} \and
            \forall r > 0:
            \left(
            \begin{array}{ll}
                q_2(r) = c_2(r,p) &\text{if }r \leq q.round\\
                q_2(r) = q(r) &\text{otherwise}\\
            \end{array}
            \right).
            \end{mathpar}

            First, $f_1$ and $f_2$ are valid for their respective
            predicates by Lemma~\ref{consValid}
            and Definition~\ref{minCons}.
            Then by validity of $f_1$ and $f_2$ and by application of
            Lemma~\ref{consValid}, we get $q_1 \in f_1$ and $q_2 \in f_2$.
            We also see that $q = q_1 \combi q_2$. Indeed, for
            $r \leq q.round$, we have $q(r) = c(r,p) = c_1(r,p) \cap c_2(r,p) =
            q_1(r) \cap q_2(r)$; and for $r > q.round$, we have
            $q(r) = q(r) \cap q(r) = q_1(r) \cap q_2 (r)$.
            Therefore $q \in \DEL_1 \combi \DEL_2$.
        \item $(\subseteq)$
            Let $q \in f_1 \combi f_2$. By Definition~\ref{defOpsStrat}
            of operations on strategies, and specifically combination,
            $\exists q_1 \in f_1, \exists q_2 \in f_2$ such that $q_1.round
            = q_2.round = q.round$ and $q = q_1 \combi q_2$.
            Since $f_1$ and $f_2$ are minimal conservative strategies of their
            respective $\DEL$s,
            by Definition~\ref{minCons}
            $\exists c_1 \in \DEL_1, \exists p_1 \in \Pi$
            such that $\forall r \leq q.round: q_1(r) = c_1(r,p_1)$; and
            $\exists c_2 \in \DEL_2, \exists p_2 \in \Pi$ such that
            $\forall r \leq q.round: q_2(r) = c_2(r,p_2)$.
            By Definition~\ref{defPrefixSym} of prefix symmetry,
            the fact that $\DEL_2$ is prefix symmetric implies that
            $\exists c'_2 \in \DEL_2$
            such that $\forall r \leq q.round: c'_2(r,p_1) = c_2(r,p_2)$.
            Hence $\forall r \leq q.round: q_2(r) = c'_2(r,p_1)$.
            By taking $c = c_1 \combi c_2$, we get
            $\forall r \leq q.round: q(r) = q_1(r) \cap q_2(r)
            = c_1(r,p_1) \cap c_2(r,p_1) = c(r,p_1)$.
            We have found $c$ and $p$ showing that $q$ is in
            the minimal conservative strategy for $\DEL_1 \combi \DEL_2$.
            \qedhere
    \end{itemize}
\end{proof}

\begin{thm}[Minimal Conservative Strategy for Succession]%
    \label{succCons}
    Let $\DEL_1, \DEL_2$ be two prefix symmetric delivered predicates,
    $f_1$ and $f_2$ the minimal
    conservative strategies for, respectively, $\DEL_1$ and $\DEL_2$.
    Then $f_1 \leadsto f_2$ is the minimal conservative strategy for
    $\DEL_1 \leadsto \DEL_2$.
\end{thm}

\begin{proof}[Proof idea]
    Since $f_1$ and $f_2$ are the minimal conservative strategies
    of $\DEL_1$ and
    $\DEL_2$, $\Nexts^C_{f_1}$ is the set of the conservative states of
    prefixes of $\DEL_1$ and $\Nexts^C_{f_2}$ is the set of the conservative
    states of prefixes of $\DEL_2$.
    Also, the states accepted by $f_1 \leadsto f_2$ are the succession
    of the states accepted by $f_1$ and the states accepted by $f_2$.
    The prefixes of $\DEL_1 \leadsto \DEL_2$ are the successions
    of prefixes of $\DEL_1$ and prefixes of $\DEL_2$ for
    the same process. Thanks to prefix symmetry,
    we can take a prefix of $\DEL_2$ and any process, and find
    a collection such that the process has that prefix.
    Therefore, the succession of prefixes for the same process are the
    same as the succession of prefixes of $\DEL_1$ and $\DEL_2$.
    Thus $\Nexts^C_{f_1 \leadsto f_2}$ is the set of conservative
    states of prefixes of $\DEL_1 \leadsto \DEL_2$,
    and is its minimal conservative strategy.
\end{proof}

\begin{proof}
    We only need to show that $f_1 \leadsto f_2$ is equal to
    Definition~\ref{minCons}.
    \begin{itemize}
        \item $(\supseteq)$ Let $q$ be a state such that
            $\exists c \in \DEL_1 \leadsto \DEL_2, \exists p \in \Pi$ such
            that $\forall r' \leq q.round: q(r') = c(r',p)$.
            By Definition~\ref{defOpsPred} of the operations on predicates, and
            specifically of succession, $\exists c_1 \in \DEL_1,
            \exists c_2 \in \DEL_2, \exists r > 0: c = c_1[1,r].c_2$.
            \begin{itemize}
                \item If $r = 0$, then $c[1,r] = c_2[1,r]$, and thus
                    $\forall r' \leq q.round: q(r') = c_2(r',p)$.
                    First, $f_2$ is valid for $\DEL_2$ by Lemma~\ref{consValid}
                    and Definition~\ref{minCons}.
                    Then the validity of $f_2$ and Lemma~\ref{consValid}
                    allow us to conclude that $q \in f_2$ and thus
                    that $q \in f_1 \leadsto f_2$.
                \item If $r > 0$, we have two cases to consider.
                    \begin{itemize}
                        \item If $q.round \leq r$, then
                            $\forall r' \leq q.round: q(r') = c_1(r',p)$
                            $f_1$ is also valid for $\DEL_1$ by Lemma~\ref{consValid}
                            and Definition~\ref{minCons}.
                            We conclude by validity of $f_1$ and application of
                            Lemma~\ref{consValid} that $q \in f_1$ and thus
                            that $q \in f_1 \leadsto f_2$.
                        \item If $q.round > r$, then $c[1,q.round]
                            = c_1[1,r].c_2[1,q.round-r]$.
                            We take $q_1$ such that
                            \begin{mathpar}
                            q_1.round = r
                            \and\text{and}\and
                            \forall r' > 0:
                            \left(
                            \begin{array}{ll}
                                q_1(r') = c_1(r',p) &\text{if }r' \leq q_1.round\\
                                q_1(r') = q(r') &\text{otherwise }\\
                            \end{array}
                            \right),
                            \end{mathpar}
                            and $q_2$ such that
                            \begin{mathpar}
                            q_2.round = q.round - r
                            \and\text{and}\and
                            \forall r' > 0:
                            \left(
                            \begin{array}{ll}
                                q_2(r') = c_2(r',p) &\text{if }r' \leq q_2.round\\
                                q_2(r') = q(r'-r) &\text{otherwise}\\
                            \end{array}
                            \right).
                            \end{mathpar}

                            Then by validity of $f_1$ and $f_2$ for their
                            respective predicates,
                            and by application of Lemma~\ref{consValid},
                            we get $q_1 \in f_1$ and $q_2 \in f_2$.

                            We also have:
			     \begin{itemize}
	                    \item $q_1.round + q_2.round = r + q.round - r = q.round$
                            \item $\forall r' \leq q_1.round = r, q(r') = c(r',p) =
                            c_1(r',p) = q_1(r')$
			     \item $\forall r' \in [q_1.round + 1, q.round], q(r') = c(r',p) =
			     c_2(r'-r,p) = q_2(r'-r) = q_2(r'-q_1.round)$
			     \item $\forall r'> q.round, q(r') = q_2 (r'-r )= q_2(r'-q_1.round)$
                            \end{itemize}
                            So Definition~\ref{defOpsStrat} gives $q = q_1 \leadsto q_2$.

                            We conclude that $q \in f_1 \leadsto f_2$.
                    \end{itemize}
            \end{itemize}
        \item $(\subseteq)$ Let $q \in f_1 \leadsto f_2$.
            By Definition~\ref{defOpsStrat} of operations for strategies,
            specifically succession, there are
            three possibilities for $q$.
            \begin{itemize}
                \item If $q \in f_1$, then by Definition~\ref{minCons} of
                  the minimal conservative strategy and minimality of $f_1$
                  for $\DEL_1$, we have
                    $\exists c_1 \in \DEL_1, \exists p_1 \in \Pi:
                    \forall r \leq q.round: q(r) = c_1(r,p_1)$.
                    Let $c_2 \in \DEL_2$. We take
                    $c = c_1[1,q.round].c_2$; we have $c \in c_1 \leadsto c_2$
                    by Definition~\ref{defOpsPred} of operations for predicates.
                    Then, $\forall r \leq q.round:
                    q(r)=c_1(r,p_1)=c(r,p_1)$.
                    We have found $c$ and $p$
                    showing that $q$ is in the minimal conservative strategy for
                    $\DEL_1 \leadsto \DEL_2$ by Definition~\ref{minCons}.
                \item If $q \in f_2$, then by Definition~\ref{minCons} of
                  the minimal conservative strategy and minimality of $f_2$
                  for $\DEL_2$, we have
                    $\exists c_2 \in \DEL_2, \exists p_2 \in \Pi:
                    \forall r \leq q.round: q(r) = c_2(r,p_2)$.
                    As $\DEL_2 \subseteq \DEL_1 \leadsto \DEL_2$
                    by Definition~\ref{defOpsPred}, thus
                    $c_2 \in \DEL_1 \leadsto \DEL_2$.
                    We have found $c$ and $p$
                    showing that $q$ is in the minimal conservative strategy for
                    $\DEL_1 \leadsto \DEL_2$ by Definition~\ref{minCons}.
                \item There are $q_1 \in f_1$ and $q_2 \in f_2$ such that
                    $q = q_1 \leadsto q_2$.
                    Because $f_1$ and $f_2$ are the minimal conservative strategies
                    of their respective $\DEL$s, then by Definition~\ref{minCons}
                    $\exists c_1 \in \DEL_1,
                    \exists p_1 \in \Pi$ such that
                    $\forall r \leq q.round: q_1(r) = c_1(r,p_1)$; and
                    $\exists c_2 \in \DEL_2, \exists p_2 \in \Pi$ such that
                    $\forall r \leq q.round: q_2(r) = c_2(r,p_2)$.
                    By Definition~\ref{defPrefixSym} of prefix symmetry,
                    the fact $\DEL_2$ is prefix symmetric implies that
                    $\exists c'_2 \in \DEL_2:
                    \forall r \leq q.round: c'_2(r,p_1) = c_2(r,p_2)$.
                    Hence $\forall r \leq q.round: q_2(r) = c'_2(r,p_1)$.
                    By taking $c = c_1[1,q_1.round].c'_2$, we have
                    $c \in c_1 \leadsto c'_2$.
                    Then $\forall r \leq q.round
                    = q_1.round + q_2.round$:
                    \[
                    \left(
                    \begin{array}{ll}
                        \begin{array}{ll}
                            q(r) & = q_1(r)\\
                                 & = c_1(r,p_1)\\
                                 & = c(r,p_1)\\
                        \end{array}
                            & \text{if }r \leq q_1.round\\ 
                        \begin{array}{ll}
                            q(r) & = q_2(r-q_1.round)\\
                                 & = c'_2(r-q_1.round,p_1)\\
                                 & = c(r,p_1)\\
                        \end{array}
                            & \text{if }r \in [q_1.round+1,q_1.round+ q_2.round]\\ 
                    \end{array}
                    \right).\]

                    We have found $c$ and $p$
                    showing that $q$ is in the minimal conservative strategy for
                    $\DEL_1 \leadsto \DEL_2$ by Definition~\ref{minCons}.
            \qedhere
            \end{itemize}
    \end{itemize}
\end{proof}

\begin{thm}[Minimal Conservative Strategy for Repetition]%
    \label{repetCons}
    Let $\DEL$ be a prefix symmetric delivered predicate,
    and $f$ be its minimal conservative strategy.
    Then $f^{\omega}$ is the minimal conservative strategy for
    $\DEL^{\omega}$.
\end{thm}

\begin{proof}
    We only have to show that $f^{\omega}$ is equal to
    Definition~\ref{minCons}.
    \begin{itemize}
        \item $(\supseteq)$ Let $q$ be a state such that
            $\exists c \in \DEL^{\omega}, \exists p \in \Pi$ such
            that $\forall r \leq q.round: q(r) = c(r,p)$.
            By Definition~\ref{defOpsPred} of operations for predicates,
            and specifically of repetition, $\exists
            {(c_i)}_{i \in \mathbb{N}^*}, \exists
            {(r_i)}_{i \in \mathbb{N}^*}$ such that
            $r_1 = 0$ and $\forall i \in \mathbb{N}^*:
            (c_i \in \DEL \land r_{i} < r_{i+1} \land
            c[r_i+1,r_{i+1}]=c_i[1,r_{i+1} - r_i])$.
            Let $k$ be the biggest integer such that $r_k \leq q.round$. We
            consider two cases.
            \begin{itemize}
                \item If $r_k = q.round$, then
                    $c[1,q.round] = c[1,r_k] = c_1[1,r_2-r_1].c_2[1,r_3-r_2] \dots c_{k-1}[1,r_k-r_{k-1}]$.
                    We take for $i \in [1,k-1]: q_i$ the state such that
                    $q_i.round = r_{i+1} -r_i$
                    and
                    $\forall r > 0$:
                    \[
                    \left(
                    \begin{array}{ll}
                        q_i(r) = c_i(r,p)
                            &\text{if }r \leq q_i.round\\
                        q_i(r) = q(r+\sum\limits_{j \in [1,i-1]} q_i.round)
                            &\text{otherwise }\\
                    \end{array}
                    \right).
                    \]

                    First, $f$ is valid for $\DEL$ by Lemma~\ref{consValid}
                    and Definition~\ref{minCons}.
                    Then by validity of $f$ and by application of
                    Lemma~\ref{consValid}, for $i \in [1,k-1]$
                    we have $q_i \in f$.
                    We see that $\forall r > 0: q(r) = (q_1 \leadsto \dots
                    \leadsto q_{k-1})(r)$. Indeed, $\forall r \in
                    [r_i + 1, r_{i+1}]: q(r) = c(r,p) = c_i(r-r_i,p) =
                    q_i(r-r_i) = q_{k-1}(r-\sum\limits_{j \in [1,k-2]} q_i.round) =
                    (q_1 \leadsto \dots \leadsto q_{k-1})(r)$; and for $r > q.round:
                    q(r) =
                    q((r-\sum\limits_{j \in [1,k-2]} q_i.round)
                    + \sum\limits_{j \in [1,k-2]} q_i.round) =
                    q_{k-1}(r-\sum\limits_{j \in [1,k-2]} q_i.round) =
                    (q_1 \leadsto \dots \leadsto q_{k-1})(r)$.

                    We conclude that $q \in f^{\omega}$.
                \item If $q.round > r_k$, we can apply the same reasoning
                    as in the previous case, the only difference being
                    $c[1,q.round] =
                    c_1[1,r_2-r_1].c_2[1,r_3-r_2] \dots
                    c_{k-1}[1,r_k-r_{k-1}].c_k[1,q.round-r_k]$.
            \end{itemize}
        \item $(\subseteq)$ Let $q \in f^{\omega}$.
            By Definition~\ref{defOpsStrat} of operations for strategies,
            $\exists q_1,q_2,\dots,q_k
            \in f: q = q_1 \leadsto q_2 \leadsto \dots \leadsto
            q_k$.
            By Definition~\ref{minCons} of the minimal conservative strategy and
            by minimality of $f$ for $\DEL$,
            $\exists c_1, c_2, \dots, c_k \in \DEL, \exists
            p_1,p_2,\dots,p_k \in \Pi:
            \forall i \in [1,k],
            q_i = \langle q_i.round, \{\langle r,j \rangle \mid
            r \leq q_i.round \land j \in c_i(r,p_i) \}$.
            By Definition~\ref{defPrefixSym} of prefix symmetry,
            $\DEL$ is prefix symmetric implies that
            $\forall i \in [2,k],\exists c'_i \in \DEL, \forall r \leq q_i.round:
            c'_i(r,p_1) = c_i(r,p_i)$. For $i = 1$, we have $c'_1 = c_1$.
            We take $c = c'_1[1,q_1.round].c'_2[1,q_2.round]\dots
            c'_{k-1}[1,q_{k-1}.round].c'_k$,
            thus $c \in c'_1 \leadsto c'_2 \leadsto \dots \leadsto c'_k$.
            Then $\forall r \leq q.round
            = \sum\limits_{i \in [1,k]} q_i.round, \exists i \in [1,k]$
            such that  $r \in \left[\sum\limits_{l \in [1,i-1]} q_l.round + 1,
            \sum\limits_{l \in [1,i]} q_l.round\right]$, and
          $q(r)  = q_i(r - \sum\limits_{l \in [1,i-1]} q_l.round)
                      = c_i(r - \sum\limits_{l \in [1,i-1]} q_l.round,p_i)
                      = c'_i(r - \sum\limits_{l \in [1,i-1]} q_l.round,p_1)
                      = c(r,p_1)$.

            We have found $c$ and $p$
            showing that $q$ is in the minimal conservative strategy for
            $\DEL^{\omega}$ by Definition~\ref{minCons}.
            \qedhere
    \end{itemize}
\end{proof}

\subsection{Computing Heard-Of Predicates of Conservative Strategies}%
\label{sec:conservative:heardof}

The analogy with oblivious strategies breaks here:
the heard-of predicate of conservative strategies is hard to compute,
as it depends in intricate ways on the delivered predicate itself.
Yet it is still possible to compute interesting information on
this HO:\@ upper bounds. These are overapproximations of the
actual HO, but they can serve for formal verification of LTL properties.
Indeed, the executions of an algorithm for the actual HO are contained
in the executions of the algorithm for any overapproximation of the
HO, and LTL properties must be true for all executions of the algorithm.
So proving the property on an overapproximation also proves it on
the actual HO\@.

\begin{thm}[Upper Bounds on HO of Minimal Conservative Strategies]%
    \label{upperBoundHO}
    Let $\DEL$, $\DEL_1$, and $\DEL_2$
    be delivered predicates containing $c_{total}$.\\
    Let $f^{cons}, f_1^{cons}, f_2^{cons}$ be their respective minimal
    conservative strategies,\\
    and $f^{obliv}, f_1^{obliv}, f_2^{obliv}$
    be their respective minimal oblivious strategies.
    Then:
    \begin{itemize}
        \item $\HO_{f_1^{cons} \cup f_2^{cons}}(\DEL_1 \cup \DEL_2)
            \subseteq \HOProd(\Nexts_{f_1^{obliv}} \cup
            \Nexts_{f_2^{obliv}})$.
        \item $\HO_{f_1^{cons} \leadsto f_2^{cons}}(\DEL_1 \leadsto \DEL_2)
            \subseteq \HOProd(\Nexts_{f_1^{obliv}} \cup
            \Nexts_{f_2^{obliv}})$.
        \item $\HO_{f_1^{cons} \combi f_2^{cons}}(\DEL_1 \combi \DEL_2)
            \subseteq \HOProd( \{ n_1 \cap n_2 \mid
            n_1 \in \Nexts_{f_1^{obliv}} \land
            n_2 \in \Nexts_{f_2^{obliv}}\})$.
        \item ${\HO_{(f^{cons})}^\omega}(\DEL^{\omega})
            \subseteq \HOProd(\Nexts_{f^{obliv}})$.
    \end{itemize}
\end{thm}

\begin{proof}
    An oblivious strategy is a conservative strategy. Therefore,
    the minimal conservative strategy always dominates
    the minimal oblivious strategy. Hence we get an upper bound
    on the heard-of predicate of the minimal conservative strategies
    by applying Theorem~\ref{oblivOpsHO}.
\end{proof}

\subsection{Domination by Conservative Strategies}%
\label{sec:conservative:domination}

Some examples above, like $\DEL_{1,\geq r}^{crash}$
(see Table~\ref{tab:examples}),
are not dominated by oblivious strategies, but are dominated
by conservative strategies. This follows from a condition on
delivered predicates and its invariance by the operations,
similar to the case of oblivious strategies.
Let's start by defining that condition and then showing that it implies domination by a conservative strategy.

\begin{defi}[Common Prefix Property]%
  \label{defCommonPref}
  Let $\DEL$ be a delivered predicate.
  $\DEL$ satisfies the \textbf{common prefix property} $\triangleq$
  $\forall c \in \DEL, \forall p \in \Pi,
  \forall r > 0, \exists c' \in \DEL, \forall q \in \Pi, \forall r' \leq r:
  c'(r',q)=c(r',p)$.
\end{defi}

The common prefix property, as its name suggests, works just like the
common round property but for a prefix. It ensures that for every prefix
of a collection in the predicate, there exists a collection where every
process shares this prefix.

\begin{thm}[Sufficient Condition of Conservative Domination]%
  \label{suffConsDom}
  Let $\DEL$ be a delivered predicate satisfying the common prefix property
  Then, there is a conservative strategy that dominates \DEL\@.
\end{thm}

\begin{proof}
  The proof is the same as for the oblivious case, except that
  the common prefix property gives us a delivered collection where everyone has
  the same exact prefix as the blocked process in the canonical execution.
\end{proof}

\begin{thm}[Domination by Conservative for Operations]%
    \label{invarSuffReac}
    Let $\DEL, \DEL_1, \DEL_2$ be delivered predicates
    with the common prefix property.
    Then $\DEL_1 \cup \DEL_2$, $\DEL_1 \combi \DEL_2$,
    $\DEL_1 \leadsto \DEL_2$, $\DEL^{\omega}$
    satisfy the common prefix property.
\end{thm}

\begin{proof}\hfill
    \begin{itemize}
        \item Let $c \in \DEL_1 \cup \DEL_2$. Thus $c \in \DEL_1$
            or $c \in \DEL_2$ by Definition~\ref{defOpsPred}.
            We fix $c \in \DEL_1$ (the other
            case is symmetric). Then for $p \in \Pi$ and $r > 0$,
            we get a $c' \in \DEL_1$
            satisfying the condition of Definition~\ref{defCommonPref}.
            Since $\DEL_1 \subseteq \DEL_1 \cup \DEL_2$,
            we get $c' \in \DEL_1 \cup \DEL_2$.
            We conclude that the common prefix property still holds
            for $\DEL_1 \cup \DEL_2$ by Definition~\ref{defCommonPref}.
        \item Let $c \in \DEL_1 \combi \DEL_2$. Then
            $\exists c_1 \in \DEL_1, \exists c_2 \in \DEL_2:
            c = c_1 \combi c_2$. For $p \in \Pi$ and $r > 0$,
            our hypothesis on $\DEL_1$ and $\DEL_2$ ensures
            that there are $c'_1 \in \DEL_1$ satisfying the
            condition of Definition~\ref{defCommonPref}
            for $c_1$ and $c_2' \in \DEL_2$ satisfying the
            condition of Definition~\ref{defCommonPref}
            for $c_2$.
            We argue that $c' = c_1' \combi c'_2$ satisfies the
            condition of Definition~\ref{defCommonPref}
            for $c$. Indeed, $\forall r' \leq r,
            \forall q \in \Pi: c(r',q) = c_1'(r',q) \combi c_2'(r',q)
            = c_1(r',p) \combi c_2(r',p) = c(r',p)$.
            We conclude that the condition of Definition~\ref{defCommonPref} still holds for
            $\DEL_1 \combi \DEL_2$.
        \item Let $c \in \DEL_1 \leadsto \DEL_2$. Since if $c \in
            \DEL_2$, the condition holds by hypothesis,
            we study the case where succession actually happens.
            Hence $\exists c_1 \in \DEL_1, \exists c_2 \in \DEL_2, \exists
            r_{change} > 0: c = c_1[1,r_{change}].c_2$.
            For $p \in \Pi$ and $r > 0$,
            our hypothesis on $\DEL_1$ and $\DEL_2$ ensures
            that there are $c'_1 \in \DEL_1$ satisfying the
            condition for $c_1$ at $r$ and $c_2' \in \DEL_2$ satisfying
            the condition for $c_2$ at $r - r_{change}$.
            We argue that $c' = c_1'[1,r_{change}].c_2'$
            satisfies the condition for $c$. Indeed,
            $\forall r' \leq r, \forall q \in \Pi$, we have:
            if $r' \leq r_{change}: c'(r',q) = c_1'(r',q)
            = c_1(r',p) = c(r',p)$; and if $r' > r_{change}:
            c'(r',q) = c_2'(r'-r_{change},q) = c_2(r'-r_{change},p)
            = c(r',p)$.
            We conclude that the condition still holds for
            $\DEL_1 \leadsto \DEL_2$.

        \item Let $c \in \DEL^{\omega}, p \in \Pi, r > 0$ and let's prove that $\exists c' \in \DEL^{\omega}, \forall q \in \Pi, \forall r' \leq r: c'(r',q)=c(r',p)$.
 $c \in \DEL^{\omega}$, so by Definition~\ref{defOpsPred}, there exist $(c_i)$ and
            $(r_i)$ such that: $r_1=0 \land \forall i \in \mathbb{N}^*:
          (c_i \in \DEL \land r_{i} < r_{i+1} \land c[r_i+1,r_{i+1}]=c_i[1,r_{i+1} - r_i])$.
          $\forall i \in \mathbb{N}^*:c_i \in \DEL$, so by hypothesis and Definition~\ref{defCommonPref}, $\forall i \in \mathbb{N}^*, \exists c'_i \in \DEL, \forall q \in \Pi, \forall r' \leq r: c'_i(r',q)=c_i(r',p)$.

Let $c' = \prod\limits_{i > 0} c'_i[1,r_{i+1}-r_i]$.
 By definition, $c' \in \DEL^{\omega}$.
 Let $q \in \Pi$ and $ r' \leq r$ and let's prove that $c'(r',q)=c(r',p)$. Let $i \in \mathbb{N}^*$ such that $r_i+1 \leq r'\leq r_{i+1}$. Then, $c'(r',q) = c'_{i}(r'-r_{i},q) = c_{i}(r'-r_{i},p) = c(r',p)$.
We conclude that the condition of Definition~\ref{defCommonPref} still holds for $\DEL^{\omega}$.
            \qedhere
    \end{itemize}
\end{proof}

\noindent
  Therefore, as long as the initial building blocks satisfy the common prefix property,
  so do the results of the operations. Thus the latter is dominated
  by its minimal conservative strategy --- a strategy that can be computed
  from the results of this section.

\section{Looking at Future Rounds}%
\label{sec:chap2future}

In the above, the dominating strategy was at most conservative:
only the past and present rounds were useful
for generating heard-of collections. Nevertheless, messages
from future rounds serve in some cases. This section provides a preliminary exploration
of the topic, with an example and a conjecture.

Let's go back to $\DEL^{loss}_1$, the delivered predicate for
at most one message loss presented in Section~\ref{sec:chap2defPdel}.
The minimal oblivious
strategy for this predicate is $f_{n-1}$. The minimal conservative strategy
is a similar one, except that when it received a message from $p$ at round $r$,
it waits for all messages from $p$ at previous rounds. However, this does not
change which messages the strategy waits for in the current round: $n-1$ messages,
because one can always deliver all the messages from the past,
and then the loss might be a message from the current round.
However, if the strategy considers messages from the next round, it
can ensure that at each round, at most one message among all processes
is not delivered on time. The strategy presented here waits for either all messages
from the current round, or for all but one messages from the current round and
all but one message from the next round.

\begin{defi}[Strategy for One Loss]
  Let $\textit{after}: Q \mapsto \mathcal{P}(\Pi)$ such that $\forall q \in Q:
  \textit{after}(q) = \{ k \in \Pi \mid \langle q.round + 1, k \rangle \in q.mes\}$.
  Then $f_{loss} \triangleq {}$
  \[ \left\{
  q \in Q ~\middle|~
  \begin{array}{ll}
      & \card(\obliv(q)) = \card(\Pi)\\
      \lor & (\card(\textit{after}(q)) = \card(\Pi)-1 \land \card(\obliv(q)) = \card(\Pi)-1) \}
  \end{array}
  \right\}\]
\end{defi}

Intuitively, this strategy is valid for $\DEL^{loss}_1$ because at each round and for each
process, only two cases exist:
\begin{itemize}
\item Either no message for this process at this round is lost, and
  it receives a message from every process;
\item Or one message for this process is lost at this round, and
    it only receives $n-1$ messages. However, all the other processes
    receive $n$ messages (because only one can be lost), thus they change round and send
    their messages for the next round. Since the one loss already
    happened, all these messages are delivered, and the original
    process eventually receives $n-1$ messages from the next round.
\end{itemize}

\begin{lem}[Validity of $f_{loss}$]%
  \label{asymValid}
  $f_{loss}$ is valid for $\DEL^{loss}_1$.
\end{lem}

\begin{proof}
  We proceed by contradiction: \textbf{Assume} $f_{loss}$ is invalid
  for \DEL$^{loss}_1$. Thus there exists
  $t \in \execs_{f_{loss}}(\DEL^{loss}_1)$ invalid.
  There is a smallest round $r$ where some process $j$ cannot change round from that point onwards.
  Let also $c$ be a delivered collection of $\DEL^{loss}_1$
  such that $t \in \execs(c)$.
  Minimality of $r$ entails that every process reaches round $r$ and thus sends
  its messages to $j$, and $c \in \DEL^{loss}_1$ entails that
  at most one of these messages can be lost. Thus $j$ eventually receives $n-1$ messages
  from round $r$.
  By hypothesis, it doesn't receive all $n$ messages, or it could change round. Thus $j$
  receives exactly $n-1$ messages from round $r$, which means that the only loss allowed
  by $\DEL^{loss}_1$ happens at round $r$.
  For $j$ to block, it must never receive $n-1$ messages from round $r+1$.
  Yet the only loss is a message to $j$; thus every other process receives $n$
  messages at round $r$, changes round, and sends its message to $j$ without loss.
  Hence $j$ eventually receives $n-1$ messages from round $r+1$.
  This \textbf{contradicts} the fact that $j$ cannot change round at this point in $t$.
\end{proof}

This strategy also ensures that at most one process per round receives
only $n-1$ messages on time --- the others must receive all the messages. This
vindicates the value of messages from future rounds for some delivered predicates,
such as the ones with asymmetry in them.

\begin{thm}[Heard-Of Characterization of $f_{loss}$]%
  \label{asymCharac}~\\
  $\HO_{f_{loss}}(\DEL^{loss}_1) = \{ h, \text{ a heard-of collection} \mid \forall r > 0, \sum\limits_{p \in \Pi} \card(\Pi \setminus ho(r,p)) \leq 1\}$.
\end{thm}

\begin{proof}\hfill
  \begin{itemize}
    \item ($\subseteq$)
  Let $\textit{ho} \in \HO_{f_{loss}}(\DEL^{loss}_1)$
  and $t \in \execs_{f_{loss}}(\DEL^{loss}_1)$ an execution of $f_{loss}$
  generating \textit{ho}.
  By definition of the executions of $f_{loss}$, processes change
  round only when they received either $n$ messages from the current round,
  or $n-1$ messages from the current round (and $n-1$ messages from the next one,
  but that is irrelevant to the heard-of predicate). Moreover, $t$ is valid by
  definition, as it generates \textit{ho}.
  Let's now assume that at least one process $j$ receives only $n-1$ messages on time
  for some round $r$ in $ho$. By definition of $f_{loss}$ and validity of $t$,
  we deduce that $j$ also received $n-1$ messages from round $r+1$ while
  it was at round $r$. Hence, every other process ended its own round $r$ before
  $j$; the only possibility is that they received $n$ messages from round $r$, because
  the alternatives require the reception of the message from $j$ at round $r+1$.
  We conclude that for each round, at most one process receives only $n-1$ messages on time,
  which can be rewritten as
  $\forall r \in \mathbb{N}^*: \sum\limits_{j \in \Pi} \card(\Pi \setminus ho(r,j)) \leq 1$.

  \item ($\supseteq$)
  Let \textit{ho} a heard-of collection over $\Pi$ such that
  $\forall r \in \mathbb{N}^*: \sum\limits_{j \in \Pi} \card(\Pi \setminus ho(r,j)) \geq 1$.
  The difficulty here is that the canonical execution of $ho$ fails to be
  an execution of $f_{loss}$: when only $n-1$ messages from the current round
  are delivered to some process $j$, then the corresponding $next_j$ for this
  round will not be allowed by $f_{loss}$.
  One way to deal with this issue is to start from the canonical execution
  of $ho$ and move these incriminating $next_j$ after the deliveries of
  $n-1$ messages from the next round, and before the deliveries of
  messages from $j$ in the next round.
  In this way, every $next_j$ will happen when either $n$ messages have been
  received from the current round, or $n-1$ messages from the current round
  and $n-1$ from the next one.
  We conclude that $\textit{ho} \in \HO_{f_{loss}}(\DEL^{loss}_1)$.
  \qedhere
  \end{itemize}
\end{proof}

\noindent
Does $f_{loss}$ dominate $\DEL^{loss}_1$? It would seem so, but proving it is
harder than for $f_{n-F}$ and $\DEL^{crash}_F$. The reason is that the common round property
of the latter allows the creation of deadlocks where every process is blocked
in the same local state, which forces any valid strategy to accept this state.
Whereas the whole reason the future serves in $\DEL^{loss}_1$ is because the latter
doesn't have this property, and thus the local state of a process constrains
the possible local states of other processes.

\begin{conj}[Domination of $f_{loss}$ on $\DEL^{loss}_1$]
$f_{loss}$ is a dominating strategy for $\DEL^{loss}_1$.
\end{conj}

This example demonstrates two important points about strategies using the future.
First, for some delivered predicates (for example $\DEL^{loss}_1$), they dominate
    conservative strategies.
  Secondly, because any additional message might influence whether the strategy allows
    a change of round or not, proof techniques based on standard executions don't
    work for these strategies.

\section{Conclusion}

In this article, we propose a formalization of the heard-of predicate to use
when studying a given operational model through the Heard-Of model. This formalization
comes with the following methodology:
\begin{itemize}
  \item Extract a delivered predicate from the informal model, either through
    direct analysis or by building the predicate from simpler building blocks through
    operations.
  \item Compute a dominating strategy for this delivered predicate, either
    by direct analysis or by combining dominating strategies for the simpler
    building blocks.
  \item Compute the heard-of predicate generated by this strategy on this
    delivered predicate.
\end{itemize}

\noindent
This result captures the most constrained predicate which can be implemented
on top of the initial model. If a round-based algorithm is proved correct
on this predicate, its correctness follows on the original model; and if no algorithm
exists for a given problem on this predicate, this entails that no round-based
algorithm solves the problem on the original model.

The obvious follow-up to this research is to tackle strategies which look
into the future, and predicates like $\DEL^{loss}_1$ that are useful for such strategies.
Doing so will require completely new proof techniques, as the ones presented
here implicitly rely on the ability to make only the messages in the past
and present rounds matter.
A more straightforward extension of this work would be the study of more
delivered predicates, both for having more building blocks to use with operations,
and to derive more heard-of predicates.

\paragraph{Acknowledgement} This work was supported by project PARDI ANR--16--CE25--0006.
The content of this article was partially published in
\emph{Characterizing Asynchronous Message-Passing Models Through
Rounds}~\cite{ShimiOPODIS18} published at OPODIS 2018, and
\emph{Derivation of Heard-Of Predicates From Elementary Behavioral
Patterns}~\cite{ShimiFORTE} published at FORTE 2020.

\bibliographystyle{plain}
\bibliography{biblio}

\begin{thebibliography}{10}

\bibitem{ArjomandiFirstRound}
Eshrat Arjomandi, Michael~J. Fischer, and Nancy~A. Lynch.
\newblock A difference in efficiency between synchronous and asynchronous
  systems.
\newblock In {\em Thirteenth Annual ACM Symposium on Theory of Computing}, STOC
  '81, pages 128--132. ACM, 1981.

\bibitem{BielyKSet}
Martin Biely, Peter Robinson, Manfred Schmid, Ulrich~Schwarz, and Kyrill
  Winkler.
\newblock Gracefully degrading consensus and k-set agreement in directed
  dynamic networks.
\newblock {\em Theoretical Computer Science}, 726:41--77, 2018.

\bibitem{BorowskyTopology}
Elizabeth Borowsky and Eli Gafni.
\newblock Generalized {FLP} impossibility result for {T}-resilient asynchronous
  computations.
\newblock In {\em Twenty-fifth Annual ACM Symposium on Theory of Computing},
  STOC '93, pages 91--100. ACM, 1993.

\bibitem{ChandraCHT}
Tushar~Deepak Chandra, Vassos Hadzilacos, and Sam Toueg.
\newblock The weakest failure detector for solving consensus.
\newblock {\em Journal of the ACM}, 43(4):685--722, July 1996.

\bibitem{CharronBostHOL}
Bernadette Charron-Bost, Henri Debrat, and Stephan Merz.
\newblock Formal verification of consensus algorithms tolerating malicious
  faults.
\newblock In {\em Stabilization, Safety, and Security of Distributed Systems},
  pages 120--134. Springer Berlin Heidelberg, 2011.

\bibitem{CharronBostApprox}
Bernadette Charron-Bost, Matthias F{\"u}gger, and Thomas Nowak.
\newblock Approximate consensus in highly dynamic networks: The role of
  averaging algorithms.
\newblock In {\em Automata, Languages, and Programming}, pages 528--539, 2015.

\bibitem{CharronBostHO}
Bernadette Charron-Bost and Andr{\'e} Schiper.
\newblock The heard-of model: computing in distributed systems with benign
  faults.
\newblock {\em Distributed Computing}, 22(1):49--71, April 2009.

\bibitem{CouloumaConsensus}
{\'E}tienne Coulouma, Emmanuel Godard, and Joseph Peters.
\newblock A characterization of oblivious message adversaries for which
  consensus is solvable.
\newblock {\em Theoretical Computer Science}, 584:80--90, 2015.

\bibitem{DragoiPsync}
Cezara Dr\u{a}goi, Thomas~A. Henzinger, and Damien Zufferey.
\newblock Psync: A partially synchronous language for fault-tolerant
  distributed algorithms.
\newblock {\em SIGPLAN Not.}, 51(1):400--415, January 2016.

\bibitem{ElradDecomp}
Tzilla Elrad and Nissim Francez.
\newblock Decomposition of distributed programs into communication-closed
  layers.
\newblock {\em Science of Computer Programming}, 2(3):155--173, 1982.

\bibitem{FisherAlgo}
Michael~J. Fischer and Nancy~A. Lynch.
\newblock A lower bound for the time to assure interactive consistency.
\newblock {\em Information Processing Letters}, 14(4):183--186, 1982.

\bibitem{FraigniaudComplexity}
Pierre Fraigniaud, Amos Korman, and David Peleg.
\newblock Towards a complexity theory for local distributed computing.
\newblock {\em Journal of the ACM}, 60(5):35:1--35:26, October 2013.

\bibitem{GafniRRFD}
Eli Gafni.
\newblock Round-by-round fault detectors (extended abstract): Unifying
  synchrony and asynchrony.
\newblock In {\em Seventeenth Annual ACM Symposium on Principles of Distributed
  Computing}, PODC '98, pages 143--152. ACM, 1998.

\bibitem{HerlihyTopology}
Maurice Herlihy and Nir Shavit.
\newblock The topological structure of asynchronous computability.
\newblock {\em Journal of the ACM}, 46(6):858--923, November 1999.

\bibitem{HutleComPredicates}
Martin Hutle and André Schiper.
\newblock Communication predicates: A high-level abstraction for coping with
  transient and dynamic faults.
\newblock In {\em 37th Annual IEEE/IFIP International Conference on Dependable
  Systems and Networks (DSN'07)}, pages 92--101, June 2007.

\bibitem{KuhnDynamic}
Fabian Kuhn and Rotem Oshman.
\newblock Dynamic networks: Models and algorithms.
\newblock {\em SIGACT News}, 42(1):82--96, March 2011.

\bibitem{MaricCutoff}
Ognjen Mari{\'{c}}, Christoph Sprenger, and David Basin.
\newblock Cutoff bounds for consensus algorithms.
\newblock In {\em Computer Aided Verification}, pages 217--237. Springer
  International Publishing, 2017.

\bibitem{NowakTopoConsensus}
Thomas Nowak, Ulrich Schmid, and Kyrill Winkler.
\newblock Topological characterization of consensus under general message
  adversaries.
\newblock In {\em 2019 ACM Symposium on Principles of Distributed Computing},
  PODC '19, pages 218--227, 2019.

\bibitem{SaksTopology}
Michael Saks and Fotios Zaharoglou.
\newblock Wait-free k-set agreement is impossible: The topology of public
  knowledge.
\newblock {\em SIAM Journal on Computing}, 29(5):1449--1483, March 2000.

\bibitem{SantoroLoss}
Nicola Santoro and Peter Widmayer.
\newblock Time is not a healer.
\newblock In {\em 6th Annual Symposium on Theoretical Aspects of Computer
  Science (STACS 89)}, pages 304--313. Springer-Verlag New York, Inc., 1989.

\bibitem{ShimiOPODIS18}
Adam Shimi, Aur{\'e}lie Hurault, and Philippe Qu{\'e}innec.
\newblock Characterizing asynchronous message-passing models through rounds.
\newblock In {\em 22nd Int'l Conf.\ on Principles of Distributed Systems
  (OPODIS 2018)}, pages 18:1--18:17, 2018.

\bibitem{ShimiFORTE}
Adam Shimi, Aurélie Hurault, and Philippe Queinnec.
\newblock Derivation of {Heard-Of} predicates from elementary behavioral
  patterns.
\newblock In {\em 40th Int'l Conf.\ on Formal Techniques for Distributed
  Objects, Components, and Systems (FORTE 2020)}, pages 133--149. Springer
  International Publishing, 2020.

\end{thebibliography}

\appendix

\section{Omitted proofs}

\begin{lem}[Lemma~\ref{stCorrect} Correctness of Standard Execution]%
  \label{app:proofStCorrect}
  Let $c$ be a delivered collection and $f$ be a strategy.
  Then $st(f,c) \in \execs_f(c)$.
\end{lem}

\begin{proof}[Proof idea]
First, the proof shows that $st(f,c)$ is indeed an execution by verifying
the three properties of Definition~\ref{defExec}. Then, it shows that
$st(f,c)$ is an execution of the delivered collection $c$
(Definition~\ref{defExecColl}: the delivered messages in $st(f,c)$ are
exactly those from $c$). Lastly the verification of the two conditions of
Definition~\ref{defExecStrat} ensures that $st(f,c)$ is an execution of
strategy $f$.
\end{proof}

\begin{proof}
  Let us first show that $st(f,c)$ is an execution by showing each
  point of Definition~\ref{defExec}.
  \begin{itemize}
    \item Delivered after sending:
      By Definition~\ref{defStdExec} of the standard execution
      there are $r-1$ transitions $next_p$ before
      the messages of $p$ sent at round $r$ are delivered.
    \item Delivered only once:
      If a message sent at round $r$ by $p$ is delivered,
      we know from the previous point that $p$ reaches
      round $r$ before the delivery. Let $r'$ be such that
      $changes_{r'}$ contains the ${(r-1)}^{\mbox{th}}$ $next_p$ of
      $st(f,c)$.
      Then the message is delivered in $dels_{r'+1}$
      by Definition~\ref{defStdExec} of the standard execution.
      For all
      $r'' > r'+1$, if there are deliveries from $p$
      in $dels_{r''}$, this entails that $next_p \in changes_{r''-1}$,
      and thus that $p$ is not anymore at round $r$. By
      definition of $dels_{r''}$, it only delivers messages
      sent at the current rounds of processes.
      We conclude that there is only one delivery of the
      message.
    \item Once stopped, forever stopped:
      By Definition~\ref{defStdExec} of the standard execution,
      if $changes_r$ contains only $stop$,
      this means that $f$ does not allow any process to change
      rounds with the currently received messages.
      By definition of $dels_r$, it only delivers messages
      from processes that changed round in $changes_{r-1}$.
      If there is some smallest $r_{stop}$ such that
      $changes_{r_{stop}} = \{stop\}$, then $dels_{r_{stop}+1} = \emptyset$.
      This means the local states of processes do not change,
      and thus $changes_{r_{stop}+1} = \{stop\}$.
      By induction, if there is some $stop$ in $st(f,c)$,
      the rest of the execution contains only $stop$ transitions.
  \end{itemize}

  Hence $st(f,c)$ is an execution.
  Next, let's show that $st(f,c)$ is an execution of $c$. By
  Definition~\ref{defExecColl}, this means the delivered messages
  are exactly those from $c$, for processes that reached the round
  where they send the message.
  By Definition~\ref{defStdExec} of the standard execution,
  all deliveries are from messages in $c$.
  By Definition~\ref{defStdExec} of the standard execution, if
  $p$ reaches round $r$, there is a smallest $r' > 0$ such
  that $round_{r'}^p = r$. This means that $dels_{r'}$ contains
  the deliveries of all the messages $(r,p,q)$ such that $p \in c(r,q)$.
  Hence $st(f,c)$ is an execution of $c$.
  Finally, we show that $st(f,c)$ is an execution of $f$. We check the two
  conditions of Definition~\ref{defExecStrat}:
  \begin{itemize}
    \item (All Nexts Allowed) By Definition~\ref{defStdExec} of the
      standard executions, changes of round only occur
      when the local state is in $f$.
    \item (Fairness) If some process is blocked forever at some round,
      this means by Definition~\ref{defStdExec} of the standard execution
      that its local state was not in $f$ for an infinite number of $changes_r$,
      and so for an infinite number of times.
  \end{itemize}
  We conclude that $st(f,c) \in \execs_f(c)$.
\end{proof}


\end{document}